\documentclass[twocolumn, aps, superscriptaddress, prl]{revtex4-1}
\usepackage{graphicx}
\usepackage{amsfonts}
\usepackage{amssymb}
\usepackage{amsthm}
\usepackage{array}
\usepackage{amsmath}
\usepackage{verbatim} 
\usepackage{hyperref}
\usepackage{color}
\usepackage{bbold}
\usepackage{epstopdf}
\usepackage{mathtools}

\newtheorem{lemma}{Lemma}

\newcommand{\ket}[1]{\left\vert#1\right\rangle}
\newcommand{\bra}[1]{\left\langle#1\right\vert}

\def\JJ#1#2{{\cal J}_{#1}^{#2}}

\def\bra#1{\langle #1|}
\def\ket#1{\left|#1 \right>}

\def\Tr{\mbox{Tr}}
\def\Lengr{{\cal L}_{Br}}

\def\Ddiss{{\cal D}_{\rm diss}}

\newcommand{\comm}[1]{}

\begin{document}
\title{Reversing Lindblad Dynamics via Continuous Petz Recovery Map}
\author{Hyukjoon Kwon}
\affiliation{QOLS, Blackett Laboratory, Imperial College London, London SW7 2AZ, United Kingdom}
\affiliation{Korea Institute for Advanced Study, Seoul 02455, South Korea}
\author{Rick Mukherjee}
\affiliation{QOLS, Blackett Laboratory, Imperial College London, London SW7 2AZ, United Kingdom}
\author{M. S. Kim}
\affiliation{QOLS, Blackett Laboratory, Imperial College London, London SW7 2AZ, United Kingdom}
\affiliation{Korea Institute for Advanced Study, Seoul 02455, South Korea}
\begin{abstract}
An important issue in developing quantum technology is that quantum states are so sensitive to noise. We propose a protocol that introduces reverse dynamics, in order to precisely control quantum systems against noise described by the Lindblad master equation. The reverse dynamics can be obtained by constructing the Petz recovery map in continuous time. By providing the exact form of the Hamiltonian and jump operators for the reverse dynamics, we explore the potential of utilizing the near-optimal recovery of the Petz map in controlling noisy quantum dynamics. While time-dependent dissipation engineering enables us to fully recover a single quantum trajectory, we also design a time-independent recovery protocol to protect encoded quantum information against decoherence. Our protocol can efficiently suppress only the noise part of dynamics thereby providing an effective unitary evolution of the quantum system.
\end{abstract}
\maketitle

The dynamics of an open quantum system is defined by the Hamiltonian of the system and its interaction with the environment. On tracing out the environment, the system undergoes a non-unitary evolution which can easily wash out coherence. For the realization of quantum technologies, it is crucial to protect the system from leaking quantum information due to its interaction with the environment. A considerable amount of effort has been made to achieve this task through developing protocols to minimize added noise \cite{LaHaye04, Wilson-Rae04, Gigan06, Chang10}, finding noise-free zones \cite{Lidar98, Viola99, Facchi04}, and correcting \cite{Shor95, Bennett96, Laflamme96, Gottesman97, Barnum02, Kitaev06, Ng10, Fowler12, Cafaro14, NielsenAndChuang} or mitigating \cite{Temme17, Li17, Endo18} errors. Especially in quantum error-correction (QEC), a universal recovery operation, the so-called Petz recovery map \cite{Petz86} has served as a useful mathematical tool to study the recovery of quantum information \cite{Barnum02, Wilde15, Sutter16, Junge18} and state discrimination protocols \cite{Holevo79, Hausladen94}. Based on the near-optimal recovery property of the map \cite{Barnum02}, approximate QEC \cite{Barnum02, Ng10, Cafaro14} has been developed. However, due to its complexity, the Petz recovery map remains in the mathematical realm, while an approach to realize the discrete version of the map was proposed very recently \cite{Gilyen20}.

In this Letter, we construct a quantum master equation which realizes the Petz recovery map in continuous time. While the Petz recovery map recovers a given quantum state after following noisy dynamics, a physical protocol to achieve the recovery map was not previously known. Our master equation identifies the reverse Hamiltonian and the jump operators that can fully reverse a quantum trajectory. We extend this to design a time-independent recovery protocol and use it to protect quantum information against decoherence. The efficient noise cancellation leads to a noiseless unitary dynamics of the encoded system. The recovery dynamics can be implemented by interacting the system with a strongly decaying ancilla.

{\it Reversing quantum master equation dynamics.---}
We focus on a Markovian open quantum dynamics described by the Lindblad equation \cite{Lindblad76} ($\hbar = 1$):
\begin{equation}
\frac{d\rho}{dt} = {\cal L} (\rho) = - i [H, \rho] +  \sum_{\mu} {\cal D}[L_\mu]  (\rho),
\end{equation}
where ${\cal D}[L_\mu] (\rho) = L_\mu \rho L_\mu^\dagger - \frac{1}{2} \{ L_\mu^\dagger L_\mu, \rho\}$. Here, $[A,B] = AB - BA$ and $\{ A,B \} = AB + BA$. A quantum state $\gamma_0$ at time $t=0$ then evolves to $\gamma_\tau = {\cal T} \left[ \exp\left( \int_0^\tau {\cal L} dt \right) \right] (\gamma_0)$ after some time $t=\tau$, where ${\cal T}$ is the time-ordering operator. We ask whether it is possible to recover the quantum state at each time from the final state $\gamma_\tau$.

\begin{figure}[t]
\includegraphics[width=1.00\linewidth]{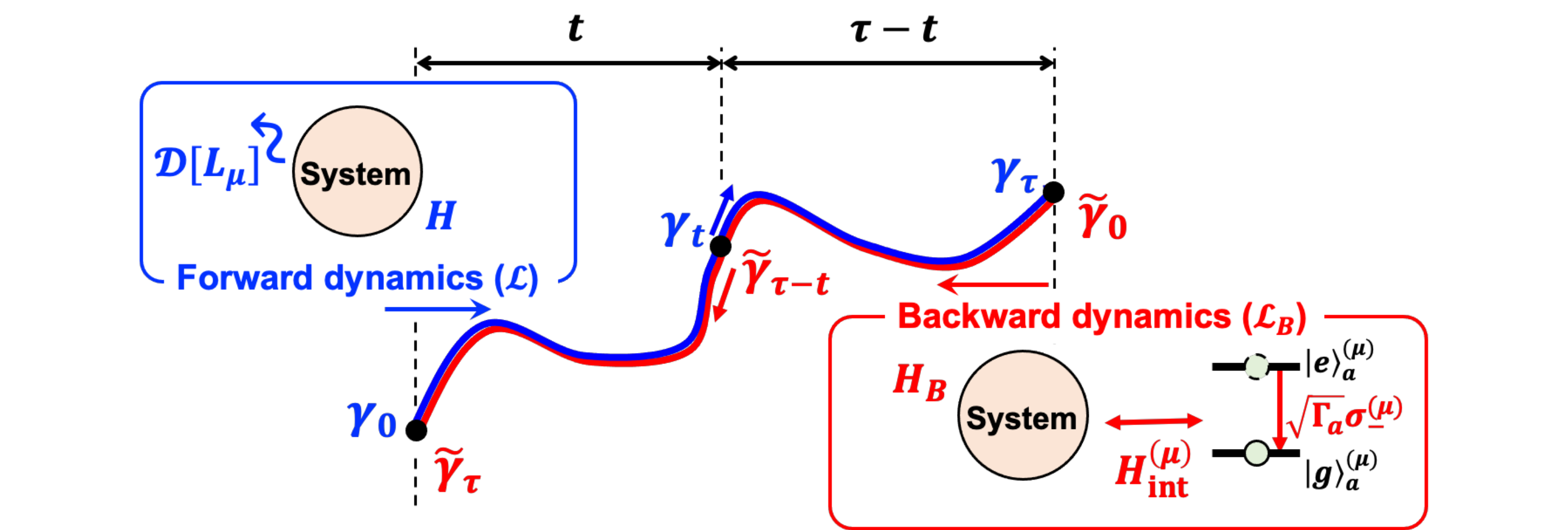}
\caption{
Forward (reverse) dynamics ${\cal L}$ (${\cal L}_B$) described by the Hamiltonian $H$ ($H_B$) and jump operators $L_\mu$ ($L_{B,\mu}$).
The reverse jump operators can be implemented by interaction Hamiltonian $H_{\rm int}^{(\mu)}$ between the system and a strongly decaying ancilla state. See the main text for details.}
\label{fig:1}
\end{figure}

To reconstruct the initial state $\gamma_0$ from the final state $\gamma_\tau$ one can adopt the Petz recovery map \cite{Petz86}, also known as the transpose channel. For a quantum channel ${\cal N}$ and a reference state $\rho$, the Petz recovery map defined as ${\cal R}_{\rho, {\cal N}} (\cdot) = \rho^{\frac{1}{2}} {\cal N}^\dagger ({\cal N}(\rho)^{-\frac{1}{2}}  (\cdot) {\cal N}(\rho)^{-\frac{1}{2}}) \rho^{\frac{1}{2}} $ recovers $\rho$ from ${\cal N}(\rho)$, i.e., 
\begin{equation}
{\cal R}_{\rho, {\cal N}} ({\cal N}(\rho)) = \rho.
\end{equation}
Such a property of the recovery map has been studied as a generalized time-reversal in various contexts, including quantum thermodynamics \cite{Faist18, Aberg18,Kwon19} and Bayesian retrodiction of quantum processes \cite{Buscemi21}. By taking a quantum channel ${\cal N}= {\cal T}\left[ \exp \left( \int_0^\tau {\cal L} dt' \right) \right]$ as the forward dynamics after time $\tau$ and the reference state $\gamma_0$, the Petz recovery map recovers $\gamma_0$ from $\gamma_\tau$, i.e., ${\cal R}_{\gamma_0, {\cal N}}(\gamma_\tau) = \gamma_0$.

Such a construction of the recovery map can be extended for any time $t \in [0, \tau]$,
based on the dynamical semigroup property of the Lindblad equation. The Petz recovery map 
in the limit of infinitesimal time interval can be described by the following Lindblad equation \cite{Kwon19}:
$$
\frac{d \rho}{d\tilde t} = {\cal L}_B ( \rho) = -i [H_B(\tilde t), \rho] + \sum_\mu {\cal D}[L_{B,\mu}(\tilde t)]( \rho).
$$
In this work, we show that the reverse 
dynamics can be expressed by the separate contributions of the forward Hamiltonian $H$ and jump operators $L_\mu$ as
\begin{equation}
\begin{aligned}
\label{eq:RevL}
H_B(\tilde t) &= -H + \sum_\mu H_C(\gamma_{\tau - \tilde t}, L_\mu)\\
L_{B,\mu}(\tilde t) &= \gamma_{\tau - \tilde t}^{\frac{1}{2}} L_{\mu}^\dagger \gamma_{\tau - \tilde t}^{-\frac{1}{2}},
\end{aligned}
\end{equation}
where the tilde indicates the backward direction and
$$
H_C(\gamma, L_\mu) = - \frac{i}{2}\sum_{\lambda, \lambda'} \left( \frac{\sqrt{\lambda}- \sqrt{\lambda'}}{\sqrt{\lambda} + \sqrt{\lambda'}} \right)\bra{\lambda} M_\mu(\gamma) \ket{ \lambda'} \ket{\lambda} \bra{\lambda'},
$$
using the eigenvalue decomposition $\gamma = \sum_{\lambda} \lambda \ket{\lambda}\bra{\lambda}$ and defining $M_\mu(\gamma) = L_{\mu}^\dagger L_{\mu} + \gamma^{-\frac{1}{2}} L_{\mu} \gamma L_{\mu}^\dagger \gamma^{-\frac{1}{2}}$.
If $\gamma$ contains zero-eigenvalues, pseudo-inverse on its support can be taken  \cite{Suppl}.
For a 
dissipation-free dynamics, the reverse dynamics takes the form ${\cal L}_B(\rho) = -i[H_B, \rho]$ with $H_B = -H$, while $L_{B,\mu}$ and $H_C$ contribute to reversing the dissipation ${\cal D}[L_\mu]$. 
The reverse dynamics fully recovers the quantum trajectory (see Fig. \ref{fig:1}), i.e.,
\begin{equation}
\tilde\gamma_{\tilde t = \tau - t} = {\cal T} \left[ e^{ \int_0^{\tau-t} {\cal L}_B d \tilde t'} \right] (\gamma_\tau) = \gamma_t, \quad \forall t \in [0, \tau],
\label{eq:rev_traj}
\end{equation}
as it satisfies ${\cal L}_B (\tilde\gamma_{\tilde t = \tau - t})  = -{\cal L} (\gamma_t)$. We note that the forward trajectory information $\gamma_{t}$ is required to construct the reverse dynamics. This can, in principle, be calculated from the initial state $\gamma_0$ and the forward dynamics ${\cal L}$, without performing state tomography. The explicit form of the reverse dynamics in Eq.~\eqref{eq:RevL} brings the abstract mathematical expression of the Petz recovery map into a physically achievable form, by identifying the Hamiltonian and jump operators.

As an illustrative example, we consider a two-level system whose dynamics is given by a Hamiltonian $H = \boldsymbol{h} \cdot \boldsymbol{\sigma}$ and a single jump operator $L = \boldsymbol{l} \cdot \boldsymbol{\sigma}$, where $\boldsymbol{\sigma} = (\sigma_x, \sigma_y, \sigma_z)$ are the Pauli operators, and $\boldsymbol{h}$ and $\boldsymbol{l}$ are real and complex vectors, respectively. Figure~\ref{fig:qubit} shows that the reverse dynamics with $H_B = \boldsymbol{h}_B \cdot \boldsymbol{\sigma}$ and $L_{B} = \boldsymbol{l}_{B} \cdot \boldsymbol{\sigma}$ obtained from Eq.~\eqref{eq:RevL} fully reverses the quantum trajectory for the initial state $\ket{0}$ such that $\sigma_z \ket{0} = \ket{0}$. This requires a temporal control of $\boldsymbol{h}_B$ and $\boldsymbol{l}_B$ with three and six independent parameters, respectively, where their closed forms and the implicit implementation of the jump operator can be found in the Supplemental Material \cite{Suppl}.

\begin{figure}[t]
\includegraphics[width=0.95\linewidth]{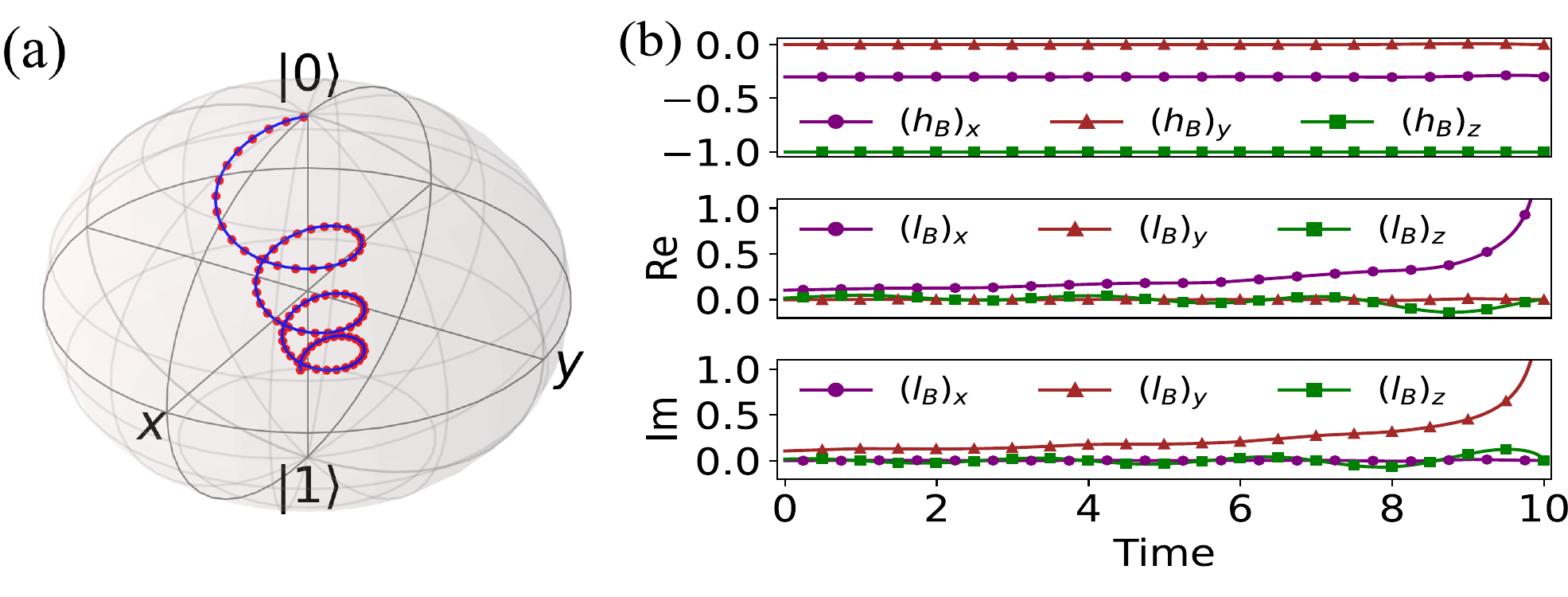}
\caption{(a) The forward (solid) and backward (dotted) trajectories of the qubit dynamics of $H = 0.3 \sigma_x + \sigma_z$, and $L = 0.4 \sigma_-$ with $\tau = 10$ and $\gamma_0 = \ket{0}\bra{0}$. (b) The Hamiltonian and jump operators for the reverse dynamics. }
\label{fig:qubit}
\end{figure}

{\it Continuous-time recovery with time-independent control.---}
To avoid the temporal control of multiple parameters which can be technically challenging, we extend the recovery protocol to have \textit{time-independent} Hamiltonian and jump operators. To this end, we consider the scenario where the forward and recovery dynamics simultaneously act on the system as,
\begin{equation}
{\cal L}_S(\rho) = \left( {\cal L} + {\cal L}_B \right) (\rho) = {\cal L}(\rho) -i [H_B, \rho] + \sum_\mu {\cal D}[ L_{B,\mu} ] (\rho),
\label{eq:rev_static}
\end{equation}
where $H_B = -H + \sum_\mu H_C(\gamma,L_\mu)$ and $L_{B,\mu} = \gamma^{\frac{1}{2}} L_{\mu}^\dagger \gamma^{-\frac{1}{2}}$ for a full-rank reference state $\gamma$. We note that $\gamma$ becomes a stationary state satisfying $\dot \gamma = {\cal L}_S(\gamma) = 0$. This can be understood as the infinitesimal time recovery ${\cal L}_B$ cancels out the effect of noise ${\cal L}$, hence trapping a fixed reference state $\gamma$ instead of reversing a trajectory $\gamma_t$.  As the time-independent formalism is less resource-intensive in its implementation, henceforth, we focus on the dynamics described by Eq.~\eqref{eq:rev_static} for the applications of our recovery protocol.


{\it Recovery of encoded quantum information.---}
While the Petz recovery map perfectly recovers the reference state, the map also enjoys the universal recovery property such that a wider spectrum of quantum states encoded in a higher-dimensional Hilbert space can be recovered close to the optimal rate \cite{Barnum02}.
In this manner, our continuous-time recovery protocol not only keeps the full-rank reference state static but also well protects any encoded quantum states, regardless of their ranks. This can be done by constructing the code space ${\cal C}$ spanned by the degenerate ground states of a Hermitian operator $Q$ and applying the reverse dynamics in Eq.~\eqref{eq:rev_static} with the following form of the reference state:
\begin{equation}
\gamma = \frac{e^{-\beta Q}}{{\rm Tr}\left[ e^{-\beta Q} \right ]}.
\label{eq:gamma}
\end{equation}
This form guarantees that $\gamma$ is full-ranked and becomes proportional to the projector onto the code space $P_{\cal C} = \sum_{\ket\psi \in {\cal C}} \ket\psi \bra \psi$ when $\beta \gg 1$, which was shown to be the reference state that efficiently preserves the code space's information \cite{Blume-Kohout08}.

\begin{figure}[t]
\includegraphics[width=1.0\linewidth]{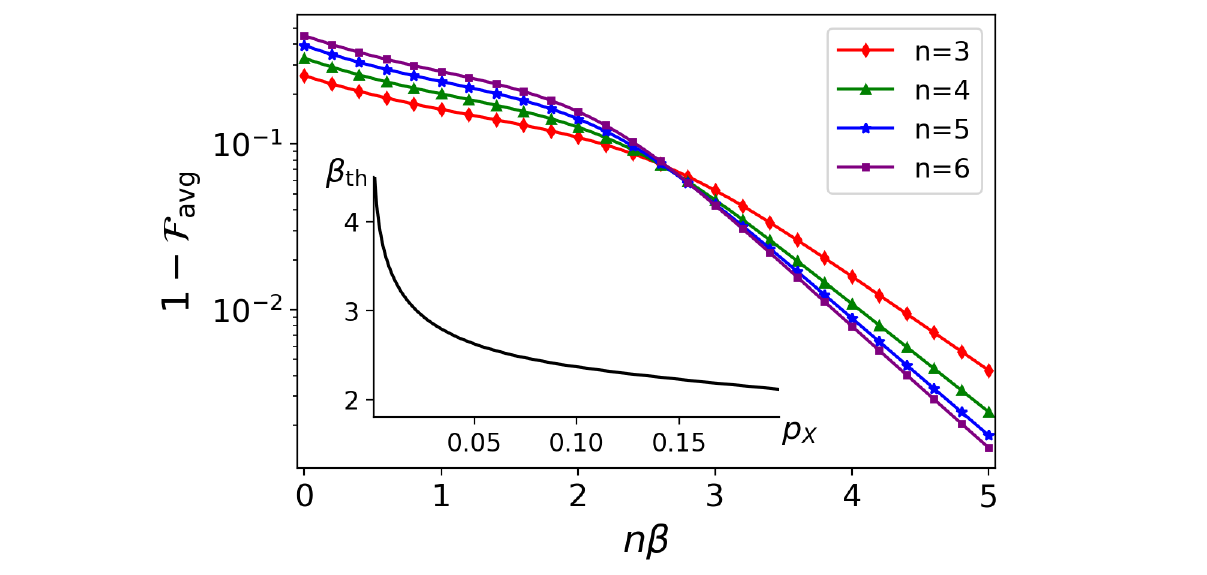}
\caption{The average infidelity $1- {\cal F}_{\rm avg}$ (main plot) and threshold value $\beta_{\rm th}$ (inset) of a single logical qubit after applying the recovery. For the main plot, the bit-flipping probability of each physical qubit is taken to be $p_X = 0.05$.}
\label{fig:spin_chain}
\end{figure}
As an example, we consider a fully-connected $n$-spin chain with $Q = -  \sum_{i > j} \sigma_z^{(i)} \sigma_z^{(j)}$, where $\sigma^{(i)}_{x,y,z}$ acts on the $i$th spin. A logical qubit is then spanned by the degenerate ground states of $Q$ as $\ket\psi = \alpha_0 \ket{0}^{\otimes n} + \alpha_1 \ket{1}^{\otimes n} \in {\cal C}$. We show that the logical qubit can be efficiently protected against the noise dynamics ${\cal L}_X = \Gamma_X \sum_{i=1}^n {\cal D}[\sigma_x^{(i)}]$, which is equivalent to independent bit-flipping errors on each qubit with probability $p_X = (1 - e^{-2\Gamma_X \tau})/2$ after time $\tau$. The recovery dynamics of Eq.~\eqref{eq:rev_static} is constructed by noting that $\gamma^{\frac{1}{2}}  \sigma_x^{(i)} \gamma^{-\frac{1}{2}} = \sigma_x^{(i)} \prod_{j \neq i} \left[ (\cosh \beta) \mathbb{1} - (\sinh \beta) \sigma_z^{(i)} \sigma_z^{(j)} \right]$ and $H_C(\gamma, \sigma_x^{(i)}) = 0$. Figure~\ref{fig:spin_chain} shows that the average fidelity ${\cal F}_{\rm avg} = \int_{\ket\psi \in {\cal C}} d\psi \bra{\psi} e^{{\cal L}_S \tau} (\ket\psi \bra\psi) \ket\psi $ increases as $\beta$ becomes larger. To obtain a larger value of $\beta$, stronger recovery is required, which can be captured by $\cosh \beta $ and $\sinh \beta$ terms in the jump operators. We also note that increasing the number of physical qubits $n$ provides a higher recovery rate when $n \beta$ exceeds a threshold ${\beta}_{\rm th}$. We observe that $\beta_{\rm th}$ becomes smaller when the bit-flipping rate $p_X$ increases (see Fig.~\ref{fig:spin_chain}), implying that a weak recovery dynamics can be effective for an intermediate noise level.

{\it Continuous recovery protocol for QEC codes.---}
The code space of the fully-connected spin can be understood using the stabilizer formalism in QEC. We further provide a general expression of the recovery protocol for any $[\![ n,k,d ]\!]$ stabilizer code, which encodes $k$-logical qubits into $n$-physical qubits with code distance $d$. Such a code construction is efficient for the Pauli-type dissipation ${\cal L} = \sum_\mu \Gamma_\mu {\cal D}[E_\mu]$ with $E_{\mu} \in \langle \sigma_x, \sigma_y, \sigma_z \rangle^n$. 
The code space ${\cal C}$ is spanned by a set of quantum states that commute with every element in a stabilizer ${\cal S}$. By noting that every state in ${\cal C}$ becomes a ground state of $Q = - \sum_{S_i \in \bar{\cal S}} S_i$ for a subset of the stabilizer $\bar{\cal S} \subset {\cal S}$, often referred to as the stabilizer Hamiltonian, the recovery dynamics becomes
\begin{equation}
{\cal L}_B = \sum_\mu \Gamma_\mu {\cal D}\bigg[E_{\mu} \prod_{S_i \in \bar{\cal S}_\mu} \left[ (\cosh \beta) \mathbb{1} - (\sinh{\beta}) S_i \right] \bigg],
\label{eq:Stab}
\end{equation}
where $\bar{\cal S}_\mu = \{ S_i \in \bar{\cal S} | \{S_i, E_\mu \} = 0 \}$. For $\beta \gg 1$ and $E_\mu$ in the correctable set, Eq.~\eqref{eq:Stab} can be interpreted as continuous syndrome measurements and corrections, which have been studied in the context of continuous QEC \cite{Paz98, Sarovar05, Pastawski11,  Lihm18} and experimentally realized in circuit QED \cite{Leghtas15, Touzard18, Gertler21}.

\begin{figure}[t]
\includegraphics[width=1.00\linewidth]{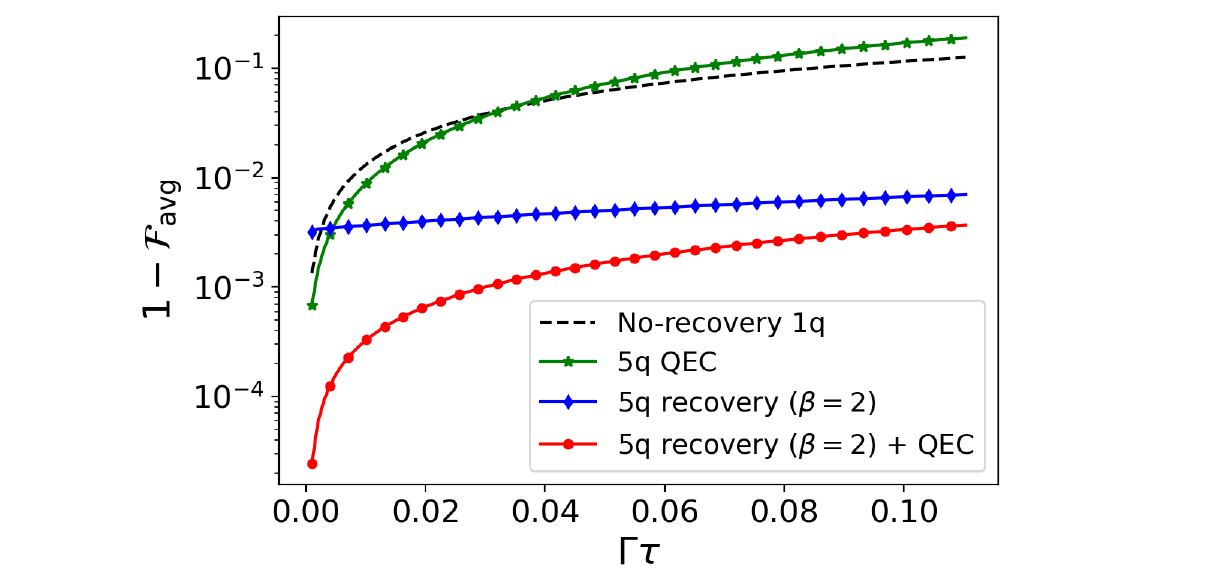}
\caption{Comparison between the average infidelities of the $[\![5,1,3]\!]$ code by applying the continuous recovery protocol and QEC. The noise is given with $\Gamma_X = \Gamma_Z = \Gamma$ and $\Gamma_{ZZ} = 0.2 \Gamma$ and the time is given in a unit of $1/\Gamma$. }
\label{fig:5q_fid}
\end{figure}

For example, we consider a noise model ${\cal L} = \sum_{i=1}^5 \left( \Gamma_X {\cal D}[\sigma_x^{(i)}] + \Gamma_Z {\cal D}[\sigma_z^{(i)}] +  \Gamma_{ZZ}{ \cal D}[\sigma_z^{(i)}\sigma_z^{(i+1)}] \right)$ and its recovery protocol applied to  the $[\![5,1,3]\!]$ code \cite{Bennett96, Laflamme96} by taking $Q = -\sum_{i=1}^5 \sigma_x^{(i)} \sigma_z^{(i+1)} \sigma_z^{(i+2)} \sigma_x^{(i+3)}$, where $\sigma_{x,z}^{(5l + j)} = \sigma_{x,z}^{(j)}$ for $l, j \in \mathbb{Z}$. Figure~\ref{fig:5q_fid} shows that the noise is suppressed by applying the recovery protocol, even for the correlated noise which cannot be directly handled with the $[\![5,1,3]\!]$ code. We also note that conventional syndrome-measurement-based QEC becomes more effective after suppressing noise via the continuous recovery protocol. For $\Gamma_X = \Gamma_Z = \Gamma$ and $\Gamma_{ZZ} = 0.2 \Gamma$, QEC is effective for all the time when the recovery is active. In contrast, without the recovery protocol, QEC is effective only for $\Gamma \tau \lesssim 0.03$ (see Fig.~\ref{fig:5q_fid}). This shows that the continuous recovery can aid QEC by reducing the noise level below the threshold, as it can be applied to more complicated stabilizer codes \cite{Kitaev06, Fowler12}.

Meanwhile, our approach is not limited to the stabilizer code or Pauli-type dissipation. From any code space ${\cal C}$ one can take $Q = -P_{\cal C}$, instead of the stabilizer Hamiltonian, to construct the recovery protocol. In the limit of $\beta \gg 1$, we obtain ${\cal L}_B \approx e^{\beta} \sum_\mu {\cal D} [ P_{\cal C} L_\mu^\dagger \left( \mathbb{1} - P_{\cal C} \right)]$, which can be understood as a continuous-time quantum jump from outside of the code space $(\mathbb{1} - P_{\cal C})$ to the code space $(P_{\cal C})$.  Such a construction is useful to achieve efficient protection of quantum states by noise-specific encoding. Other noise models, including amplitude damping, with various types of the code space construction are discussed in the Supplemental Material \cite{Suppl}.

For both recovery protocols based on stabilizer and general code spaces, the reference state $\gamma$ does not need to be prepared to implement the recovery dynamics as it can be fully determined only from $Q$ and $\beta$. The average error is governed by the factor $\beta$, which can be interpreted as the inverse temperature by noting that the most states remain at the ground state, i.e., the code space at the low temperature ($\beta \gg 1$).  In this limit, the code space ${\cal C}$ becomes a decoherence-free subspace of the reverse dynamics ${\cal L}_B$, and at the same time, ${\cal L}_B$ continuously brings quantum states outside the code space back to the code space. From the similarity appearing in the stationary state in Eq.~\eqref{eq:gamma}, the recovery map can be compared to self-correcting quantum memory \cite{Dennis02, Alicki10, Bravyi13}. 
However, we highlight that our recovery dynamics in Eq.~\eqref{eq:rev_static} meets the quantum detailed balance relation \cite{Alhambra17} even without the Davies map condition \cite{Davies74}.

The recovery protocol not only provides a longer lifetime of logical qubits, but also can be utilized to simulate noise-free dynamics of a logical Hamiltonian $H_L(t)$ acting on the code space. In the following dynamics,
\begin{equation}
\dot \rho= -i [H_L(t), \rho] + {\cal L} (\rho) + {\cal L}_B(\rho),
\label{eq:eff_U}
\end{equation}
the recovery ${\cal L}_B$ continuously cancels out the noise ${\cal L}$, thereby a quantum state $\rho_{\cal C}$ in the code space evolves effectively unitarily, $\dot\rho_{\cal C} \approx -i [H_L(t), \rho_{\cal C}]$. Figure~\ref{fig:5q_H} shows the effective unitary dynamics of the $[\![5,1,3]\!]$ code under Eq.~\eqref{eq:eff_U}, where the logical operations are defined as $X_L = \bigotimes_{i=1}^5 \sigma_x^{(i)}$ and $Z_L = \bigotimes_{i=1}^5 \sigma_z^{(i)}$. We note that the same recovery protocol is applied to any initial states in the code space.
\begin{figure}[t]
\includegraphics[width=1.0\linewidth]{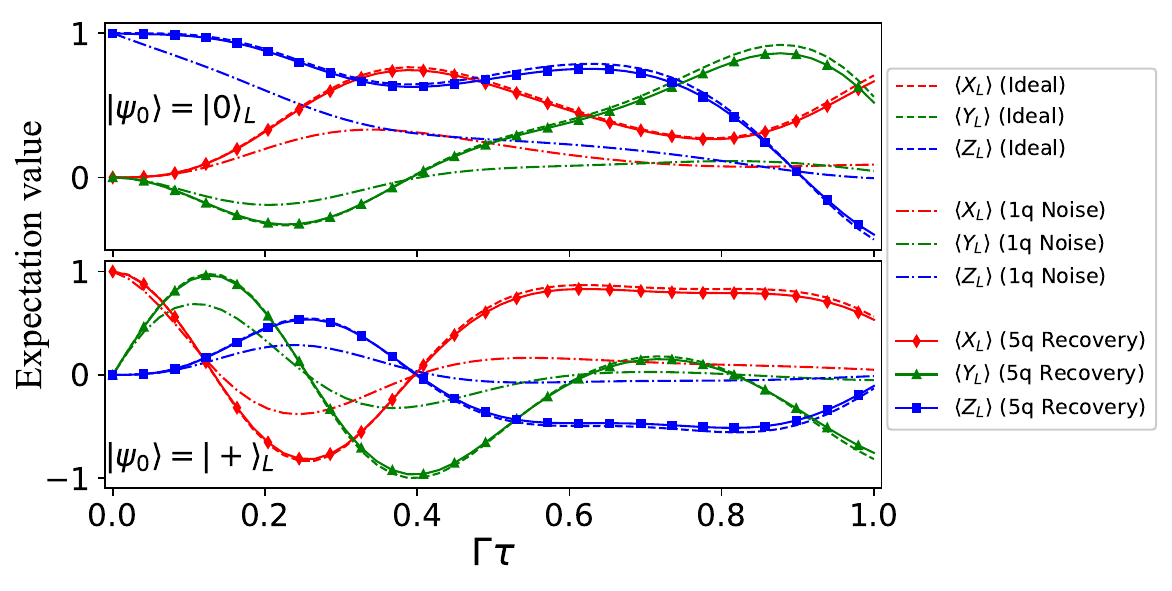}
\caption{Simulating the dynamics of the $[\![5,1,3]\!]$ code under a time-dependent Hamiltonian $H_L(t) = 3 \sin(5t) X_L + 6 \cos(2t) Z_L$. The initial states are prepared in $\ket{0}_L$ (top) and $\ket{+}_L = (\ket{0}_L + \ket{1}_L)/\sqrt{2}$ (bottom), respectively, and the same recovery operation ${\cal L}_B$ with $\beta = 2$ is applied to both cases without QEC. The noise model is the same as in Fig.~\ref{fig:5q_fid}.}
\label{fig:5q_H}
\end{figure}

{\it Implementation of the reverse dynamics and resource analysis.---}
The dissipation engineering for the continuous recovery map can be done by introducing a strongly decaying ancillary system. Assuming a two-level ancilla state, the jump operators $L_{B,\mu}$ can be implemented throughout the system-ancilla interaction,
\begin{equation}
H_{\rm int}^{(\mu)}= \frac{1}{2} \left( L_{B,\mu}^\dagger \otimes \sigma_-^{(\mu)} + L_{B,\mu} \otimes \sigma_+^{(\mu)} \right),
\end{equation}
and the dissipation of the ancilla $\Gamma_a {\cal D} [\sigma_-^{(\mu)}]$. In the limit of strong dissipation, one can adiabatically eliminate the excited states of the ancilla so that the system's effective dynamics can be well approximated as $ (1/\Gamma_a){\cal D}  [L_{{B},\mu}]$  \cite{Verstraete09}.  Thus, by tuning the Hamiltonian of the total system to be $H_{sa} = H_B +  \sqrt{\Gamma_a} \sum_\mu H_{\rm int}^{(\mu)}$, the dynamics
\begin{equation}
{\cal L}_{sa} (\rho_{sa}) = -i [H_{sa}, \rho_{sa}] +  \Gamma_a \sum_\mu {\cal D}[\sigma_-^{(\mu)}] (\rho_{sa}),
\end{equation}
leads to the effective system dynamics ${\cal L}_B$ after tracing out the ancilla system. The strength of the engineered dissipation can be expressed in terms of a dimensionless parameter $\| H_{sa} \|^2/(\Gamma_a \Gamma_\mu)$ \cite{Lihm18}, where $\Gamma_\mu = \| L_\mu \|^2$ is the intrinsic dissipation strength of the system. In order to engineer the dissipation successfully, we need to satisfy $\Gamma_a \gg \| H_{sa} \| \gg \sqrt{\Gamma_a \Gamma_\mu}$. This condition can be achieved in ultracold and circuit QED systems. In particular with Rydberg systems, both the strength of the system-ancilla interaction and the decay rate of an effective two-level system can be controlled either by choosing a specific principal quantum number or by applying the concept of  Rydberg dressing \cite{Mukherjee16, Zeiher16}. This flexibility can provide a whole range of values for $\| H_{sa} \|\sim$ MHz to GHz while $\Gamma_{a},\Gamma_{\mu} \sim$ kHz to tens MHz \cite{Saffman10}.

For the qubit system considered in Fig.~\ref{fig:qubit}, $H_{\rm int}^{(\mu)}$ requires two-qubit Pauli operations between the system and ancilla. One way to construct a Hamiltonian with universal two-body interactions would be to use superconducting qubits \cite{Kapit15}. Another possible route would be to use ultracold systems. Examples of engineering general two-body spin interactions involve cold atoms~\cite{Hung16}, ions \cite{Kim10}, and anisotropic interactions between Rydberg atoms \cite{Glaetzle15} or polar molecules \cite{ Wall}. On the other hand, characterizing the noise and engineering of jump operators of a multi-qubit system might be a challenging task to achieve the continuous recovery protocol.
Nevertheless, recent progress on dissipation engineering \cite{Verstraete09, Morigi15, Horn18, Damanet19} could lead to resolving technical challenges. 

The total number of qubits required to achieve the reverse dynamics for $n$-physical qubits is $n + n_a$, where $n_a$ is the number of ancilla qubits to implement the reverse jump operators. We note that $n_a$ equals the number of jump operators describing noise acting on $n$-physical qubits, no matter local or correlated. For a local noise model, $n_a$ linearly scales with $n$, which is comparable to the number of ancilla qubits required for syndrome measurements in the standard QEC. When noise at each physical qubit has pairwise correlations with at most $\ell$ other qubits, $n_a$ scales no more than ${\cal O}(n \ell)$. In addition, the time-independent recovery protocol has the advantage that it does not require system controls conditioned on the syndrome measurement outcomes, as well as additional classical computation for diagnosing the errors from the syndromes.

{\it Remarks.---}
We have shown that Lindblad dynamics can be reversed by constructing the Petz recovery map in continuous time. We have provided an explicit form of the Hamiltonian and jump operators as well as a possible route for physical realization of such dynamics throughout the adiabatic elimination technique. As an application, we have shown that the continuous recovery protocol can be designed for QEC, which provides a high recovery rate of encoded quantum information against noisy environment.

Our recovery protocol can be applied to implementing a noiseless quantum gate \cite{Sun20, Lau21} and dynamical quantum noise canceling \cite{Tsaang10}, which might be feasible for small scale noisy quantum devices \cite{Preskill18}. This will open a new possibility to utilize the near-optimal recovery property of the Petz recovery map not only in approximate QEC \cite{Barnum02,Ng10, Cafaro14} and quantum communication \cite{Beigi16}, but also revealing fundamental physics in quantum thermodynamics \cite{Faist18,Aberg18, Kwon19} and the AdS/CFT correspondence \cite{Cotler19, Jia20}. Our formalism is limited to Markovian noise but can be further generalized to a noise model with time-dependent jump operators. An interesting future research would be exploring whether this formalism can be extended to non-Markovian dynamics, such as $1/f$ noise in superconducting qubits \cite{Simmonds04, Koch07}.

\begin{acknowledgments}
This work is supported by the KIST Open Research Program, the QuantERA ERA-NET within the EU Horizon 2020 Programme, and the UK Hub in Quantum Computing and Simulation, part of the UK National Quantum Technologies Programme with funding from UKRI EPSRC Grants No. EP/T001062/1 and No. EP/R044082/ 1. H.K. is supported by the KIAS Individual Grant No. CG085301 at Korea Institute for Advanced Study (KIAS). M. S. K. acknowledges the KIAS visiting professorship for support.
\end{acknowledgments}

\newpage
\widetext

\section{Supplemental Material}

\section{I. Constructing Continuous Petz recovery map}
\subsection{A. Obtaining the Lindblad equation for reverse dynamics}
We briefly sketch how the continuous Petz recovery map can be expressed in a Lindblad equation,
$$
\frac{d \rho}{d\tilde t} = {\cal L}_B ( \rho) = -i [H_B(\tilde t), \rho] + \sum_\mu {\cal D}[L_{B,\mu}(\tilde t)]( \rho),
$$
by following Ref.~\cite{Kwon19}. Here, the Hamiltonian $H_B$ and the jump operators $L_{B,\mu}$ are given as
\begin{equation}
\begin{aligned}
\label{eq:supp_RevL}
H_B(\tilde t) &= \left. -\frac{1}{2}  \gamma_t^{-\frac{1}{2}} \left( K - i \partial_t \right) \gamma_t^{\frac{1}{2}}  + {\rm h.c.} \right |_{t = \tau - \tilde t}\\
L_{B,\mu}(\tilde t) &= \left. \gamma_t^{\frac{1}{2}} L_{\mu}^\dagger \gamma_t^{-\frac{1}{2}} \right |_{t = \tau - \tilde t},
\end{aligned}
\end{equation}
where the backward time is defined as $\tilde t = \tau - t$ and $K = H - (i/2) \sum_{\mu} L_{\mu}^\dagger L_{\mu}$.  Let us consider a quantum channel that reconstructs $\gamma_0$ from $\gamma_\tau$. For a given quantum channel ${\cal N}$ that maps a quantum state $\rho$ to ${\cal N}(\rho)$, the Petz recovery map is a quantum channel defined as
\begin{equation}
{\cal R}_{\rho,{\cal N}} = \JJ{\rho}{\frac{1}{2}} \circ {\cal N}^\dagger \circ \JJ{ {\cal N}(\rho)}{-\frac{1}{2}},
\end{equation}
where $\JJ{A}{\alpha} (B) = A^{\alpha} B \left(A^{\dagger}\right)^{ \alpha^*}$ is so-called the rescaling operation. The Petz map fully recovers the reference state $\rho$ from ${\cal N}(\rho)$, i.e., ${\cal R}_{\rho, {\cal N}} ({\cal N}(\rho)) = \rho$. By taking ${\cal N} = {\cal T}\left[ e^{ \int_0^\tau {\cal L} dt} \right]$ to be a quantum channel describing the quantum Liouville dynamics after time $\tau$ and $\gamma_0$ to be the reference state, the Petz recovery map ${\cal R}_{\gamma_0, {\cal N}}$ leads to ${\cal R}_{\gamma_0, {\cal N}} (\gamma_\tau) = \gamma_0$. We then construct a recovery map at each time $t$ by dividing the Petz recovery map into multiple steps: 
$$
{\cal R}_{\gamma_0, {\cal N}} =\JJ{\gamma_{0}}{\frac{1}{2}} \circ {\cal T} \left[ e^{\int_0^\tau {\cal L}^\dagger dt } \right] \circ \JJ{\gamma_{\tau}}{-\frac{1}{2}} = \left( \JJ{\gamma_{0}}{\frac{1}{2}} \circ {\cal T} \left[ e^{ \int_0^{\Delta t}{\cal L}^\dagger dt } \right] \circ \JJ{\gamma_{\Delta t}}{-\frac{1}{2}} \right) \circ \cdots \circ \left( \JJ{\gamma_{\tau - \Delta t}}{\frac{1}{2}} \circ \left[ e^{ \int_{\tau - \Delta t}^\tau{\cal L}^\dagger dt } \right] \circ \JJ{\gamma_{\tau}}{-\frac{1}{2}} \right).
$$
It is important to note that $\left( \JJ{\gamma_{n \Delta t}}{\frac{1}{2}} \circ {\cal T} \left[ e^{ \int_{n\Delta t}^{(n+1)\Delta t} {\cal L}^\dagger dt } \right] \circ \JJ{\gamma_{(n+1) \Delta t}}{-\frac{1}{2}} \right) $ becomes the Petz recovery map for each time interval $t \in [n\Delta t, (n+1) \Delta t]$. By taking the limit $\Delta t \rightarrow 0$ at time $t = n\Delta t$,
$$
\left( \JJ{\gamma_{n \Delta t}}{\frac{1}{2}} \circ {\cal T} \left[ e^{ \int_{n\Delta t}^{(n+1)\Delta t} {\cal L}^\dagger dt } \right] \circ \JJ{\gamma_{(n+1) \Delta t}}{-\frac{1}{2}} \right) \xrightarrow{\Delta t \rightarrow 0} \left( \JJ{\gamma_{t}}{\frac{1}{2}} \circ e^{ {\cal L}^\dagger dt }\circ \JJ{\gamma_{t+ dt}}{-\frac{1}{2}} \right) = e^{ {\cal L}_B dt + {\cal O}(dt^2)},
$$
where 
$$
\begin{aligned}
{\cal L}_B (\bullet) dt &= -i \left( \gamma_t^\frac{1}{2} H  \gamma_{t+dt}^{-\frac{1}{2}} \bullet  \gamma_{t+dt}^{-\frac{1}{2}}  \gamma_t^\frac{1}{2}  - \gamma_t^\frac{1}{2} \gamma_{t+dt}^{-\frac{1}{2}} \bullet  \gamma_{t+dt}^{-\frac{1}{2}}  H \gamma_t^\frac{1}{2} \right) dt  \\
&\quad +\sum_\mu \left[ \gamma_t^\frac{1}{2} L_\mu^\dagger \gamma_{t+dt}^{-\frac{1}{2}} \bullet \gamma_{t+dt}^{-\frac{1}{2}} L_\mu \gamma_t^\frac{1}{2}- \frac{1}{2} \gamma_t^\frac{1}{2} L_\mu^\dagger L_\mu \gamma_{t+dt}^{-\frac{1}{2}} \bullet \gamma_{t+dt}^{-\frac{1}{2}}  \gamma_t^\frac{1}{2} - \frac{1}{2} \gamma_t^\frac{1}{2}  \gamma_{t+dt}^{-\frac{1}{2}} \bullet \gamma_{t+dt}^{-\frac{1}{2}} L_\mu^\dagger L_\mu \gamma_t^\frac{1}{2} \right] dt.
\end{aligned}
$$
By defining $\tau - n \Delta t = \tau- t = \tilde t$ and expressing the formula to the first order of $dt$ as
$$
{\cal L}_B (\bullet) dt = \left( -i [H_B(\tilde t), \bullet ] + \sum_\mu {\cal D}[L_{B,\mu}(\tilde t)]( \bullet) \right) dt+ {\cal O}(dt^2),
$$
we obtain the reverse Lindblad equation in Eq.~\eqref{eq:supp_RevL}. The reverse dynamics fully recovers the forward trajectory $\gamma_t$ as $\tilde \gamma_{\tilde t = \tau - t} = \gamma_t$ and ${d \tilde\gamma_{\tilde t}} / {d \tilde t}|_{\tilde t = \tau - t} =  {\cal L}_B (\tilde\gamma_{\tilde t})|_{\tilde t = \tau - t} = -{\cal L}(\gamma_{t}) = - d \gamma_t / dt$ for all $t \in [0, \tau]$.  A detailed derivation of Eq.~\eqref{eq:supp_RevL} can be found in \cite{Kwon19}.

\subsection{B. Alternative form of the recovery Hamiltonian}
In this section, we show the one of our main results that the recovery Hamiltonian can be expressed in an alternative form,
\begin{equation}
\begin{aligned}
H_B(\tilde t) &= -H + \sum_\mu H_C(\gamma_{\tau - \tilde t}, L_\mu)\\
&= - H - \left. \frac{i}{2} \sum_\mu \sum_{\lambda_t, \lambda'_t} \left( \frac{\sqrt{\lambda_t}- \sqrt{\lambda'_t}}{\sqrt{\lambda_t} + \sqrt{\lambda'_t}} \right)\bra{\lambda_t} M_\mu(\gamma_t) \ket{ \lambda'_t} \ket{\lambda_t} \bra{\lambda'_t} \right|_{t = \tau - \tilde t},
\end{aligned}
\label{eq:Supp_1}
\end{equation}
where $\gamma_t = \sum_{\lambda_t} \lambda_t \ket{\lambda_t}\bra{\lambda_t}$ and $M_\mu(\gamma) = L_{\mu}^\dagger L_{\mu} + \gamma^{-\frac{1}{2}} L_{\mu} \gamma L_{\mu}^\dagger \gamma^{-\frac{1}{2}}$. We derive this from Eq.~\eqref{eq:supp_RevL} as follows:
\begin{proof} In the proof, we shall work on the forward time $t = \tau - \tilde t$, which leads to a simpler expression of $\gamma_{\tau - \tilde t} = \gamma_t$ in Eq.~\eqref{eq:Supp_1}. Let us start with the following identity,
$$
{\cal L} (\gamma_t) = \partial_t \left( \gamma_t^{\frac{1}{2}} \gamma_t^{\frac{1}{2}} \right)= \gamma_t^{\frac{1}{2}} \left( \partial_t  \gamma_t^{\frac{1}{2}} \right) +   \left( \partial_t  \gamma_t^{\frac{1}{2}} \right) \gamma_t^{\frac{1}{2}}.
$$
We then express the time-derivative term as
$$
\partial_t  \gamma_t^{\frac{1}{2}} = \sum_{\lambda_t,\lambda'_t} \left( \frac{\bra{\lambda_t} {\cal L} (\gamma_t) \ket{\lambda'_t}}{\sqrt{\lambda_t} + \sqrt{\lambda'_t}}\right) \ket{\lambda_t} \bra{\lambda'_t}.
$$
Here, for mathematical simplicity, we assume that $\gamma_t$ is a full rank matrix. However, even if $\gamma_t$ contains zero-eigenvalues, we can still restrict the recovery dynamics to be acting on the subspace spanned by the eigenstates of $\gamma_t$ with non-zero eigenvalues $\lambda_t$. In this case, the recovery Hamiltonian and jump operators can be redefined to be $H_B \rightarrow \Pi_{\gamma_t} H_B \Pi_{\gamma_t}$ and $L_{B,\mu} \rightarrow \Pi_{\gamma_t} L_{B,\mu} \Pi_{\gamma_t}$, where $\Pi_{\gamma_t} = \sum_{\lambda_t \neq 0} \ket{\lambda_t} \bra{\lambda_t}$ is the projector onto the subspace spanned by the eigenstates of $\gamma_t$ having non-zero eigenvalues. We then decompose $\bra{\lambda_t} {\cal L} (\gamma_t) \ket{\lambda'_t}$ using the master equation as
\begin{equation}
\begin{aligned}
\bra{\lambda_t} {\cal L} (\gamma_t) \ket{\lambda'_t} &= \bra{\lambda_t} \left( -i [H ,\gamma_t] + \sum_\mu L_\mu\gamma_t L_\mu^\dagger - \frac{1}{2}  \sum_\mu L_\mu^\dagger L_\mu \gamma_t -  \frac{1}{2} \gamma_t \sum_\mu L_\mu^\dagger L_\mu \right) \ket{\lambda'_t}\\
&= i \left( \lambda_t - \lambda'_t \right) \bra{\lambda_t} H  \ket{\lambda_t'} +  \bra{\lambda_t} \left( \sum_\mu L_\mu\gamma_t L_\mu^\dagger \right) \ket{\lambda_t'} - \frac{1}{2} \left( \lambda_t + \lambda'_t \right) \bra{\lambda_t} \left( \sum_\mu L_\mu^\dagger L_\mu \right) \ket{\lambda_t'} .
\end{aligned}
\label{eq:Supp_Lind}
\end{equation}
We then split the reverse Hamiltonian in Eq.~\eqref{eq:supp_RevL} into three parts as
$$
\begin{aligned}
H_B(\tilde t)
&= -\frac{1}{2}  \gamma_t^{-\frac{1}{2}} \left( H - \frac{i}{2} \sum_\mu L_\mu^\dagger L_\mu  - i \partial_t \right) \gamma_t^{\frac{1}{2}}  + {\rm h.c.}\\
&= \left( -\frac{1}{2} \gamma_t^{-\frac{1}{2}} H \gamma_t^{\frac{1}{2}} + {\rm h.c.} \right) + \left( \frac{i}{4} \gamma_t^{-\frac{1}{2}}  \sum_\mu L_\mu^\dagger L_\mu \gamma_t^{\frac{1}{2}} + {\rm h.c} \right) +  \left( \frac{i}{2} \gamma_t^{-\frac{1}{2}} \partial_t \gamma_t^{\frac{1}{2}} + {\rm h.c.} \right),
\end{aligned}
$$
where the first two terms can be expressed in terms of the eigenstate $\{ \ket{\lambda_t} \}$ as
\begin{equation}
\begin{aligned}
-\frac{1}{2} \gamma_t^{-\frac{1}{2}} H \gamma_t^{\frac{1}{2}} + {\rm h.c.} 
 &= -\frac{1}{2} \sum_{\lambda_t, \lambda_t'} \left( \sqrt{\frac{\lambda'_t}{\lambda_t}} +  \sqrt{\frac{\lambda_t}{\lambda'_t}} \right) \bra{\lambda_t} H  \ket{\lambda_t'} \ket{\lambda_t} \bra{\lambda'_t}\\
\frac{i}{4} \gamma_t^{-\frac{1}{2}} \sum_\mu  L_\mu^\dagger L_\mu \gamma_t^{\frac{1}{2}} + {\rm h.c}
&= \frac{i}{4} \sum_\mu \sum_{\lambda_t, \lambda_{t'}} \left( \sqrt{\frac{\lambda'_t}{\lambda_t}} -  \sqrt{\frac{\lambda_t}{\lambda'_t}} \right)  \bra{\lambda_t} L_\mu^\dagger L_\mu  \ket{\lambda_t'} \ket{\lambda_t} \bra{\lambda'_t}.
\end{aligned}
\label{eq:Suppl_HH}
\end{equation}
The last term can  be written as
\begin{equation}
\begin{aligned}
\frac{i}{2} \gamma_t^{-\frac{1}{2}} \partial_t \gamma_t^{\frac{1}{2}} + {\rm h.c.} 
&= \frac{i}{2} \sum_{\lambda_t, \lambda_{t'}} \left( \frac{1}{\sqrt{\lambda_t}} -  \frac{1}{\sqrt{\lambda'_t}}  \right) \left( \frac{\bra{\lambda_t} {\cal L} (\gamma_t) \ket{\lambda'_t}}{\sqrt{\lambda_t} + \sqrt{\lambda'_t}}\right) \ket{\lambda_t} \bra{\lambda'_t}\\
&= \frac{i}{2} \sum_{\lambda_t, \lambda_{t'}} \left( \frac{\sqrt{\lambda'_t} - \sqrt{\lambda_t}}{\sqrt{\lambda_t \lambda'_t} (\sqrt{\lambda_t} + \sqrt{\lambda'_t}) } \right) \bra{\lambda_t} {\cal L} (\gamma_t) \ket{\lambda'_t} \ket{\lambda_t} \bra{\lambda'_t}\\
&= -\frac{1}{2}  \sum_{\lambda_t, \lambda_{t'}} (\lambda_t - \lambda'_t) \left( \frac{\sqrt{\lambda'_t} - \sqrt{\lambda_t}}{\sqrt{\lambda_t \lambda'_t} (\sqrt{\lambda_t} + \sqrt{\lambda'_t}) } \right) \bra{\lambda_t} H  \ket{\lambda_t'} \ket{\lambda_t} \bra{\lambda'_t} \\
&\quad +\frac{i}{2}  \sum_\mu \sum_{\lambda_t, \lambda_{t'}} \left( \frac{\sqrt{\lambda'_t} - \sqrt{\lambda_t}}{\sqrt{\lambda_t \lambda'_t} (\sqrt{\lambda_t} + \sqrt{\lambda'_t}) } \right) \bra{\lambda_t} L_\mu \gamma_t L_\mu^\dagger  \ket{\lambda_t'} \ket{\lambda_t} \bra{\lambda'_t}\\
&\quad - \frac{i}{4} \sum_\mu \sum_{\lambda_t, \lambda_{t'}} (\lambda_t + \lambda'_t) \left( \frac{\sqrt{\lambda'_t} - \sqrt{\lambda_t}}{\sqrt{\lambda_t \lambda'_t} (\sqrt{\lambda_t} + \sqrt{\lambda'_t}) } \right) \bra{\lambda_t} L_\mu^\dagger L_\mu  \ket{\lambda_t'} \ket{\lambda_t} \bra{\lambda'_t} ,
\end{aligned}
\label{eq:Suppl_gamma_dt}
\end{equation}
where the last expression can be obtained by substituting Eq.~\eqref{eq:Supp_Lind}. By collecting the terms containing $H$ in Eq.~\eqref{eq:Suppl_HH} and Eq.~\eqref{eq:Suppl_gamma_dt}, we obtain
\begin{equation}
\begin{aligned}
&- \frac{1}{2} \sum_{\lambda_t, \lambda'_t} \left [ \left( \sqrt{\frac{\lambda'_t}{\lambda_t}} +  \sqrt{\frac{\lambda_t}{\lambda'_t}} \right) +
(\lambda_t - \lambda'_t) \left( \frac{\sqrt{\lambda'_t} - \sqrt{\lambda_t}}{\sqrt{\lambda_t \lambda'_t} (\sqrt{\lambda_t} + \sqrt{\lambda'_t}) } \right)
 \right]\bra{\lambda_t} H  \ket{\lambda_t'} \ket{\lambda_t} \bra{\lambda'_t} \\
&\qquad = - \sum_{\lambda_t, \lambda'_t} \bra{\lambda_t} H  \ket{\lambda_t'} \ket{\lambda_t} \bra{\lambda'_t}\\
&\qquad = -H .
\end{aligned}
\label{eq:Supp_HRev1}
\end{equation}
Similarly, we the second term involving $L_\mu \gamma_t L_\mu^\dagger$ in Eq.~\eqref{eq:Suppl_gamma_dt} can be written as
\begin{equation}
\begin{aligned}
&\frac{i}{2} \sum_\mu  \sum_{\lambda_t, \lambda'_t} \left( \frac{\sqrt{\lambda'_t} - \sqrt{\lambda_t}}{\sqrt{\lambda_t}\sqrt{\lambda'_t} (\sqrt{\lambda_t} +\sqrt{\lambda'_t})}  \right) \bra{\lambda_t} L_\mu\gamma_t L_\mu^\dagger \ket{\lambda_t'} \ket{\lambda_t} \bra{\lambda'_t}\\
&\qquad = -\frac{i}{2} \sum_\mu  \sum_{\lambda_t, \lambda'_t} \left( \frac{\sqrt{\lambda_t} - \sqrt{\lambda'_t}}{\sqrt{\lambda_t} +\sqrt{\lambda'_t}}  \right) \bra{\lambda_t} \gamma_t^{-\frac{1}{2}} L_\mu\gamma_t L_\mu^\dagger \gamma_t^{-\frac{1}{2}} \ket{\lambda_t'} \ket{\lambda_t} \bra{\lambda'_t}.
\end{aligned}
\label{eq:Supp_HRev3}
\end{equation}
Finally, we can combine the terms involving $\sum_\mu L_\mu^\dagger L_\mu$ in Eq.~\eqref{eq:Suppl_HH} and Eq.~\eqref{eq:Suppl_gamma_dt} to get
\begin{equation}
\begin{aligned}
& \frac{i}{4} \sum_\mu  \sum_{\lambda_t, \lambda'_t} \left [ \left( \sqrt{\frac{\lambda'_t}{\lambda_t}} -  \sqrt{\frac{\lambda_t}{\lambda'_t}} \right) -  (\lambda_t + \lambda'_t) \left( \frac{\sqrt{\lambda'_t} - \sqrt{\lambda_t}}{\sqrt{\lambda_t \lambda'_t} (\sqrt{\lambda_t} + \sqrt{\lambda'_t}) } \right) \right]\bra{\lambda_t} L_\mu^\dagger L_\mu \ket{\lambda_t'} \ket{\lambda_t} \bra{\lambda'_t} \\
&\qquad = - \frac{i}{2} \sum_\mu \sum_{\lambda_t, \lambda'_t}  \left( \frac{\sqrt{\lambda_t} - \sqrt{\lambda'_t}}{\sqrt{\lambda_t} +\sqrt{\lambda'_t}}\right) \bra{\lambda_t} L_\mu^\dagger L_\mu \ket{\lambda_t'} \ket{\lambda_t} \bra{\lambda'_t}.
\end{aligned}
\label{eq:Supp_HRev2}
\end{equation}

Adding all the contributions from Eqs.~\eqref{eq:Supp_HRev1}, \eqref{eq:Supp_HRev3}, and \eqref{eq:Supp_HRev2} leads to the desired formula in Eq.~\eqref{eq:Supp_1}. We additionally note that $M_\mu$ can also be expressed as 
$$
M_\mu(\gamma_{\tau - \tilde t}) = L_{\mu}^\dagger L_{\mu} + L_{B,\mu}^\dagger(\tilde t) L_{B,\mu}( \tilde t).
$$
\end{proof}

\subsection{C. Effective Hamiltonian dynamics via time-dependent dissipation control}
In the previous section, we have shown that the recovery dynamics,
$$
\begin{aligned}
{\cal L}_B(\rho) &=  -i [H_B(\tilde t), \rho] + \sum_\mu {\cal D}[L_{B,\mu}(\tilde t)]( \rho)\\
&=-i [H, \rho] - i \sum_\mu [  H_C(\gamma_{\tau - \tilde t}, L_\mu), \rho] + \sum_\mu {\cal D}[L_{B,\mu} (\tilde t)](\rho),
\end{aligned}
$$
leads to a full reverse trajectory of $\gamma_t$. Now we show that one can modify the recovery dynamics to cancel only the dissipation part ${\cal D}[L_\mu]$ to obtain the noise-free Hamiltonian dynamics of $H$. To this end, we construct the following recovery dynamics,
$$
{\cal L}'_B (\rho) = -i [H'_B(\tilde t), \rho ] + \sum_\mu {\cal D}[L'_{B,\mu}(\tilde t)] (\rho),
$$
in terms of the (modified) recovery Hamiltonian and jump operators
\begin{equation}
\begin{aligned}
H'_B(\tilde t) &= \sum_\mu U_{\tau- \tilde t} H_C(\gamma_{\tau - \tilde t}, L_\mu) U_{\tau - \tilde t}^\dagger\\
L'_{B,\mu}(\tilde t) &= U_{\tau - \tilde t} L_{B,\mu} (\tilde t) U_{\tau - \tilde t}^\dagger,
\end{aligned}
\label{eq:Suppl_noRevH}
\end{equation}
where $U_t = {\cal T} \left(e^{-i \int_0^t H  dt'} \right)$. By defining ${\cal U}_t (\rho) = U_t \rho U^\dagger_t$ and ${\cal U}^\dagger_t (\rho) = U^\dagger_t \rho U_t$, we express the evolution of the quantum state after following the reverse dynamics for time $\tau$ as
$$
\begin{aligned}
\tilde \gamma'_{\tilde t = \tau} &= {\cal T} \left( e^{\int_0^\tau {\cal L}'_B d\tilde t'}\right) (\gamma_\tau) \\
&={\cal U}_\tau \circ \left[ {\cal U}^\dagger_\tau \circ  {\cal T} \left( e^{\int_{\tau - \Delta t}^\tau {\cal L}_B' d\tilde t'}\right) \circ {\cal U}_{\tau - \Delta t} \right] \circ \left[ {\cal U}^\dagger_{\tau - \Delta t} \circ {\cal T} \left( e^{\int_{\tau - 2\Delta t}^{\tau-\Delta t} {\cal L}_B' d\tilde t'}\right) \circ {\cal U}_{\tau - 2\Delta t} \right] \circ \cdots \circ \left[ {\cal U}^\dagger_{\Delta t} \circ {\cal T} \left( e^{\int_0^{\Delta t} {\cal L}_B' d\tilde t'}\right) \right] (\gamma_\tau)\\
&={\cal U}_\tau \circ {\cal T}\left( e^{\int_0^\tau {\cal L}_B d \tilde t'} \right)  (\gamma_\tau)\\
&={\cal U}_\tau (\gamma_0).
\end{aligned}
$$
The third line in the equation above can be obtained as follows. We first note that each block of the evolution can be expressed as 
\begin{equation}
\begin{aligned}
{\cal U}^\dagger_{t_n} \circ {\cal T} \left( e^{\int_{t_n - \Delta t}^{t_n} {\cal L}'_B d\tilde t'}\right) \circ {\cal U}_{t_n - \Delta t}  
&= {\cal U}^\dagger_{t_n} \circ {\cal T} \left( e^{\int_{t_n - \Delta t}^{t_n} {\cal L}'_B d\tilde t'}\right) \circ {\cal U}_{t_n} \circ {\cal U}^\dagger_{\Delta t},
\end{aligned}
\label{eq:Suppl_no_rev_H2}
\end{equation}
where $t_n = n \Delta t$. By taking $\Delta t \rightarrow 0$, we have ${\cal U}^\dagger_{\Delta t}(\rho) \rightarrow \rho + i [H , \rho] \Delta t$ and $t_n$ becomes a continuous time $\tilde t$ of the reverse dynamics. Finally, by noting that $( {\cal U}^\dagger_{\tau - \tilde t} \circ {\cal L}'_B \circ {\cal U}_{\tau - \tilde t} ) (\rho) = -i [H_C(\tilde t), \rho] + \sum_\mu {\cal D}[L_{B,\mu}(\tilde t)] (\rho)$, Eq.~\eqref{eq:Suppl_no_rev_H2} becomes exactly same as the reverse dynamics ${\cal L}_B$.

\section{II. Reverse dynamics of a two-level system}
\subsection{A. Explicit form of the reverse dynamics}
Suppose that a two-level system's dynamics ${\cal L}(\rho)= - i [\boldsymbol{h} \cdot \boldsymbol{\sigma}, \rho] + \sum_\mu {\cal D}[\boldsymbol{l}_\mu \cdot \boldsymbol{\sigma}](\rho)$ is given by the forward Hamiltonian $H = \boldsymbol{h} \cdot \boldsymbol{\sigma}$ and jump operators $L_\mu = \boldsymbol{l}_\mu \cdot \boldsymbol{\sigma}$ with a real vector $\boldsymbol{h}$, complex vectors $\boldsymbol{l}_\mu$, and the vector of Pauli matrices $\boldsymbol{\sigma} = (\sigma_x, \sigma_y, \sigma_z)$. We derive an explicit form of the reverse dynamics ${\cal L}_B (\rho)= - i [\boldsymbol{h}_B \cdot \boldsymbol{\sigma}, \rho] + \sum_\mu {\cal D}[\boldsymbol{l}_{B,\mu} \cdot \boldsymbol{\sigma}](\rho)$ directly from Eq.~\eqref{eq:supp_RevL}:
$$
\begin{aligned}
\boldsymbol{h}_B &= -\boldsymbol{h}  + \left( \frac{\cosh^2(x)}{ 2 \cosh^2(x/2) }\right)  \sum_{\mu} \left[ {\rm Re} [(\boldsymbol{r}(t) \cdot \boldsymbol{l}_{\mu}) (\boldsymbol{r}(t)  \times \boldsymbol{l}_{\mu}^*)] - \left( \frac{ \sinh^2(x/2)}{\cosh(x)}\right) [\boldsymbol{r}(t) \times  (i \boldsymbol{l}_{\mu}^* \times \boldsymbol{l}_{\mu})  ] \right]\\
\boldsymbol{l}_{B,\mu} &= \boldsymbol{l}_{\mu}^* - \left( \frac{\cosh^2(x)}{ 2 \cosh^2(x/2) }\right) \big[ \boldsymbol{r}(t) \times (\boldsymbol{r}(t) \times \boldsymbol{l}_{\mu}^*) \big] + i \cosh(x) [ \boldsymbol{r}(t) \times \boldsymbol{l}_{\mu}^* ],
\end{aligned}
$$
where $x(t) = \tanh^{-1}|\boldsymbol r(t) |$ for the forward trajectory of quantum state $\gamma_t = \frac{1}{2} \left[ \mathbb{1} + \boldsymbol{r}(t) \cdot \boldsymbol{\sigma} \right]$.

\begin{proof}
We start with expressing a two-level quantum state as 
$$
\gamma_t = \frac{1}{2} \left[ \mathbb{1} + \boldsymbol{r}(t) \cdot \boldsymbol{\sigma} \right]
= \frac{1}{2} \left[ \mathbb{1} - \tanh(x) (\boldsymbol{n}(t) \cdot \boldsymbol{\sigma}) \right] 
 = \frac{1}{Z(t)} e^{-x (\boldsymbol{n}(t) \cdot \boldsymbol{\sigma})},
$$
where $\boldsymbol{n}(t) = - \boldsymbol{r}(t) / |\boldsymbol{r}(t)|$ and $Z(t) = \Tr[e^{- x (\boldsymbol{n}(t) \cdot \boldsymbol{\sigma})}] =\cosh(x)$. Note that $x(t)$, $\boldsymbol{n}(t)$, and $Z(t)$ are time-dependent. For simplicity, we shall use the expressions $x$, $\boldsymbol{r}$, $\boldsymbol{n}$ and $Z$, but note that they are still time-dependent quantities. Let us start with the following Lemma:
\begin{lemma} For any real number $x$, unit vector $\boldsymbol{n}$, and complex vector $\boldsymbol{v}$, 
\begin{equation} 
e^{- x (\boldsymbol{n} \cdot \boldsymbol\sigma)/2} \left( \boldsymbol{v} \cdot \boldsymbol{\sigma} \right) e^{ x (\boldsymbol{n} \cdot \boldsymbol\sigma)/2} 
= \left[ \boldsymbol{v} - 2 \sinh^2( x/2)  ( \boldsymbol{n} \times (\boldsymbol{n} \times \boldsymbol{v} )) - i \sinh(x) (\boldsymbol{n}\times \boldsymbol{v}) \right] \cdot \boldsymbol\sigma.
\end{equation}
\label{Lemma:Suppl1}
\end{lemma}
\begin{proof}
We note that $[(\boldsymbol{n} \cdot \boldsymbol{\sigma}),  \left( \boldsymbol{v} \cdot \boldsymbol{\sigma} \right)] = 2 i (\boldsymbol{n} \times \boldsymbol{v}) \cdot \boldsymbol{\sigma} = 2i (\boldsymbol{C}_{\boldsymbol{n}} \boldsymbol{v}) \cdot \boldsymbol{\sigma}$, where $\boldsymbol{C}_{\boldsymbol{n}}  = \left( \begin{matrix} 0 & -n_z & n_y \\ n_z & 0 & -n_x \\ -n_y & n_x & 0\end{matrix} \right)$ satisfies $\boldsymbol{C}_{\boldsymbol{n}}^3 = -\boldsymbol{C}_{\boldsymbol{n}}$ and $\boldsymbol{C}_{\boldsymbol{n}} \boldsymbol{v} = \boldsymbol{n} \times \boldsymbol{v}$. From the Baker-Campbell-Hausdorff formula, we then obtain,
\begin{equation}
\begin{aligned}
e^{- x (\boldsymbol{n} \cdot \boldsymbol\sigma)/2} \left( \boldsymbol{v} \cdot \boldsymbol{\sigma} \right) e^{ x (\boldsymbol{n} \cdot \boldsymbol\sigma)/2}  
&= (\boldsymbol{v} \cdot \boldsymbol{\sigma}) + \left(- \frac{x}{2}\right) [(\boldsymbol{n} \cdot \boldsymbol{\sigma}), (\boldsymbol{v} \cdot \boldsymbol{\sigma})] + \frac{1}{2!} \left(-\frac{x}{2}\right)^2 [(\boldsymbol{n} \cdot \boldsymbol{\sigma}), [(\boldsymbol{n} \cdot \boldsymbol{\sigma}), (\boldsymbol{v} \cdot \boldsymbol{\sigma})]] + \cdots  \\
&= \left[ \left( \mathbb{1} + \left( - i x \boldsymbol{C}_{\boldsymbol{n}} \right)  + \frac{1}{2!} \left( - i x \boldsymbol{C}_{\boldsymbol{n}} \right)^2 + \cdots \right)  \boldsymbol{v}  \right] \cdot \boldsymbol{\sigma}\\
&= \left[ \left( \mathbb{1} -i \sum_{k=0}^\infty \frac{x^{2k+1}}{(2k+1)!} \boldsymbol{C}_{\boldsymbol{n}}  - \sum_{k=1}^\infty \frac{x^{2k}}{(2k)!}\boldsymbol{C}_{\boldsymbol{n}}^2 \right)  \boldsymbol{v}  \right] \cdot \boldsymbol{\sigma}\\
&= \left[ \boldsymbol{v} - i \sinh{x} \boldsymbol{C}_{\boldsymbol{n}} - (\cosh{x}-1) \boldsymbol{C}_{\boldsymbol n}^2\right] \cdot \boldsymbol{\sigma}\\
&= \left[ \boldsymbol{v} - 2 \sinh^2( x/2)  ( \boldsymbol{n} \times (\boldsymbol{n} \times \boldsymbol{v} )) - i \sinh(x) (\boldsymbol{n}\times \boldsymbol{v}) \right] \cdot \boldsymbol\sigma,
\end{aligned}
\end{equation}
which completes the proof.
\end{proof}

From Lemma~\ref{Lemma:Suppl1}, we directly obtain the jump operators of the reverse dynamics as
\begin{equation}
\begin{aligned}
L_{B,\mu}
&= e^{- x (\boldsymbol{n} \cdot \boldsymbol{\sigma}) /2}  \left( \boldsymbol{l}_{\mu}^* \cdot \boldsymbol{\sigma} \right)e^{ x (\boldsymbol{n} \cdot \boldsymbol{\sigma}) /2}  = \left[ \boldsymbol{l}_{\mu}^* - 2 \sinh^2( x/2)  ( \boldsymbol{n} \times (\boldsymbol{n} \times \boldsymbol{l}_{\mu}^* )) - i \sinh(x) (\boldsymbol{n} \times \boldsymbol{l}_{\mu}^* ) \right] \cdot \boldsymbol\sigma.
\end{aligned}
\end{equation}
By using the fact that $\boldsymbol{n} = -\coth(x) \boldsymbol{r}$, we obtain the desired formula of the jump operators,
\begin{equation}
\begin{aligned}
\boldsymbol{l}_{B,\mu} &= \left[ \boldsymbol{l}_{\mu}^* - 2 \sinh^2( x/2)  \coth^2(x)( \boldsymbol{r} \times (\boldsymbol{r} \times \boldsymbol{l}_{\mu}^* )) - i \sinh(x) (-\coth(x)) (\boldsymbol{r} \times \boldsymbol{l}_{\mu}^* ) \right] \cdot \boldsymbol\sigma\\
&= \boldsymbol{l}_{\mu}^* - \left( \frac{\cosh^2(x)}{ 2 \cosh^2(x/2) }\right) \big[ \boldsymbol{r} \times (\boldsymbol{r} \times \boldsymbol{l}_{\mu}^*) \big] + i \cosh(x) [ \boldsymbol{r} \times \boldsymbol{l}_{\mu}^* ].
\end{aligned}
\end{equation}

Next, we derive the Hamiltonian for the reverse dynamics. Similarly to the jump operators, each component of the reverse Hamiltonian can be calculated as
\begin{equation}
H_B(\tilde t)|_{\tilde t = \tau - t} = -\frac{1}{2}  \gamma_t^{-\frac{1}{2}} \left( H  - (i/2) \sum_{\mu} L_\mu^\dagger L_\mu - i \partial_t \right) \gamma_t^{\frac{1}{2}}  + {\rm h.c.}.
\label{ep:Suppl_HB}
\end{equation}
Here, we use Eq.~\eqref{eq:supp_RevL} instead of Eq.~\eqref{eq:Supp_1} as the derivative of $\gamma_t^{1/2}$ can be directly calculated for a two-level system. The first term of Eq.~\eqref{ep:Suppl_HB} can be directly obtained from Lemma~\ref{Lemma:Suppl1} as
\begin{equation}
\begin{aligned}
- \frac{1}{2}\gamma_t^{\frac{1}{2}} H  \gamma_t^{-\frac{1}{2}} + {\rm h.c.} 
= - \left[ \boldsymbol{h}  - 2 \sinh^2(x/2)  ( \boldsymbol{n} \times (\boldsymbol{n} \times \boldsymbol{h}  )) \right] \cdot \boldsymbol{\sigma}.
\end{aligned}
\label{eq:Suppl_h1}
\end{equation}
The second term of Eq.~\eqref{ep:Suppl_HB}, $\displaystyle \left( -\frac{i }{4} \right) \gamma_t^{\frac{1}{2}} \left(  \sum_{\mu} L_\mu^\dagger L_\mu \right) \gamma_t^{-\frac{1}{2}} + {\rm h.c.}$, can also be calculated as follows. We note that $ \sum_{\mu} L_\mu^\dagger L_\mu =  \sum_{\mu} (\boldsymbol{l}^*_{\mu} \cdot \boldsymbol{\sigma}) (\boldsymbol{l}_{\mu} \cdot \boldsymbol{\sigma}) = \sum_{\mu} \left[ (\boldsymbol{l}_{\mu}^* \cdot \boldsymbol{l}_{\mu}) \mathbb{1} + i (\boldsymbol{l}_{\mu}^* \times \boldsymbol{l}_{\mu}) \cdot \boldsymbol{\sigma} \right]$, where the term proportional to $(i \mathbb{1})$ is cancelled out when taking Hermitian conjugate and $i (\boldsymbol{l}_{\mu}^* \times \boldsymbol{l}_{\mu})$ is a real vector. We then obtain
\begin{equation}
\begin{aligned}
\left( -\frac{i }{4} \right) \gamma_t^{\frac{1}{2}} \left(  \sum_{\mu} L_\mu^\dagger L_\mu \right) \gamma_t^{-\frac{1}{2}} + {\rm h.c.} 
= \left(-\frac{1}{2} \right) \sinh(x)  \left[ \boldsymbol{n} \times i \sum_{\mu} (\boldsymbol{l}_{\mu}^* \times \boldsymbol{l}_{\mu}) \right] \cdot \boldsymbol{\sigma}.
\end{aligned}
\label{eq:Suppl_h2}
\end{equation}
Finally, we evaluate the time derivative term of Eq.~\eqref{ep:Suppl_HB}, $\displaystyle - \frac{i }{2} \left( \frac{d \gamma_t^{\frac{1}{2}}}{dt}\right) \gamma_t^{-\frac{1}{2}} + {\rm h.c.}$. We start with the following expression:
$$
\begin{aligned}
\left( \frac{d \gamma_t^{\frac{1}{2}}}{dt}\right) \gamma_t^{-\frac{1}{2}} 
&= \frac{d}{dt} \left[ \frac{1}{\sqrt{Z}} e^{- x(\boldsymbol{n} \cdot \boldsymbol{\sigma})/2} \right] \sqrt{Z}  e^{ x(\boldsymbol{n} \cdot \boldsymbol{\sigma})/2} \\
&= - \frac{ \dot{Z}}{ 2 Z } +  \frac{d}{dt} \left[ e^{- x(\boldsymbol{n} \cdot \boldsymbol{\sigma})/2} \right] e^{- x(\boldsymbol{n} \cdot \boldsymbol{\sigma})/2}\\
&= -\frac{\dot{Z}}{2 Z} +\left[ \int_0^1 d\alpha  e^{-\alpha x (\boldsymbol{n} \cdot \boldsymbol{\sigma})/2} \frac{d}{dt} \left[ \left( -\frac{ x }{2} \right) (\boldsymbol{n} \cdot \boldsymbol{\sigma}) \right] e^{-(1-\alpha) x (\boldsymbol{n} \cdot \boldsymbol{\sigma})/2} \right]  e^{ x (\boldsymbol{n} \cdot \boldsymbol{\sigma})/2} \\
&= -\frac{\dot{Z}}{2 Z} + \left( -  \frac{\dot x}{2} \right)  (\boldsymbol{n} \cdot \boldsymbol{\sigma}) +  \left( -  \frac{ x}{2} \right) \int_0^1 d\alpha  e^{-\alpha x (\boldsymbol{n} \cdot \boldsymbol{\sigma})/2}  (\dot{\boldsymbol{n}} \cdot \boldsymbol{\sigma}) e^{ \alpha x (\boldsymbol{n} \cdot \boldsymbol{\sigma})/2}.
\end{aligned}
$$
As the first two hermitian terms are cancelled out after taking the Hermitian conjugate of $\left( -\frac{i}{2} \right) \frac{d \gamma_t^{\frac{1}{2}}}{dt}$, we have
\begin{equation}
\begin{aligned}
- \frac{i}{2} \left( \frac{d \gamma_t^{\frac{1}{2}}}{dt}\right) \gamma_t^{-\frac{1}{2}} + {\rm h.c.} = \left( \frac{x}{2} \right) \left[ \int_0^1 d\alpha \sinh(\alpha x)  (\boldsymbol{n} \times \dot{\boldsymbol{n}}) \right] \cdot \boldsymbol\sigma = \sinh^2 ( x /2 )(\boldsymbol{n} \times \dot{\boldsymbol{n}}) \cdot \boldsymbol{\sigma}.
\end{aligned}
\label{eq:Suppl_dt}
\end{equation}
From the master equation, $\displaystyle \dot{\gamma_t} = {\cal L}  (\gamma_t) = -i [H , \gamma_t] + \sum_\mu L_\mu \gamma_t L_\mu^\dagger - \frac{1}{2} \{ L_\mu^\dagger L_\mu , \gamma_t\}$, we note that
\begin{equation}
\begin{aligned}
\dot{\boldsymbol{r}} = 2 (\boldsymbol{h}  \times \boldsymbol{r}) + \sum_{\mu} \left[  (\boldsymbol{l}_{\mu} \cdot \boldsymbol{r}) \boldsymbol{l}_{\mu}^* + (\boldsymbol{l}_{\mu}^* \cdot \boldsymbol{r}) \boldsymbol{l}_{\mu} - 2 (\boldsymbol{l}_{\mu}^* \boldsymbol{l}_{\mu}) \boldsymbol{r} - 2 i (\boldsymbol{l}_{\mu}^* \times \boldsymbol{l}_{\mu}) \right].
\end{aligned}
\end{equation}
We additionally note that
$ \boldsymbol{n} \times \dot{\boldsymbol{n}} = \boldsymbol{n} \times \left( \frac{\dot{\boldsymbol{r}}}{|\boldsymbol{r}|} \right)$, since $\boldsymbol n \times \boldsymbol n = 0$.
Equation~\eqref{eq:Suppl_dt} then becomes
\begin{equation}
\begin{aligned}
- \frac{i}{2} \left( \frac{d \gamma_t^{\frac{1}{2}}}{dt}\right) \gamma_t^{-\frac{1}{2}} + {\rm h.c.}
= 2 \sinh^2 ( x /2 ) \left[ \boldsymbol{n} \times (\boldsymbol{h} \times \boldsymbol{n} ) + \sum_\mu {\rm Re} [(\boldsymbol{l}_{\mu} \cdot \boldsymbol{n}) (\boldsymbol{n} \times \boldsymbol{l}_{\mu}) ]  - \boldsymbol{n} \times \left( \frac{ i \boldsymbol{l}_{\mu}^* \times \boldsymbol{l}_{\mu}}{|\boldsymbol{r}|} \right) \right] \cdot \boldsymbol{\sigma} .
\end{aligned}
\label{eq:Suppl_h3}
\end{equation}
Finally, substituting $\boldsymbol{n} = -\coth(x) \boldsymbol{r}$ and combining Eqs.~\eqref{eq:Suppl_h1}, \eqref{eq:Suppl_h2}, and \eqref{eq:Suppl_h3}, we obtain the reverse Hamiltonian $\boldsymbol{h}_B$. Note that the first term of Eq.~\eqref{eq:Suppl_h3} cancels out the second term of Eq.~\eqref{eq:Suppl_h1} so that the term depending on the forward Hamiltonian $-\boldsymbol{h} $ becomes fully decoupled from the quantum state's trajectory $\boldsymbol{r}$ as predicted in Eq.~\eqref{eq:Supp_1}.
\end{proof}

From Eq.~\eqref{eq:Suppl_noRevH}, one can also reverse noisy quantum dynamics of the two-level system to obtain the effective unitary dynamics. Figure~\ref{fig:Suppl_no_rev_H} shows the reversed trajectory and control parameters of the Hamiltonian and jump operators. Here, the forward dynamics with $H  = 0.3 \sigma_x + \sigma_z$ and $L = 0.4 \sigma_-$ and the initial state of $\gamma_0 = \ket{0}\bra{0}$ are chosen to be the same as in the main text.

\begin{figure}[ht]
\includegraphics[width=0.25\linewidth]{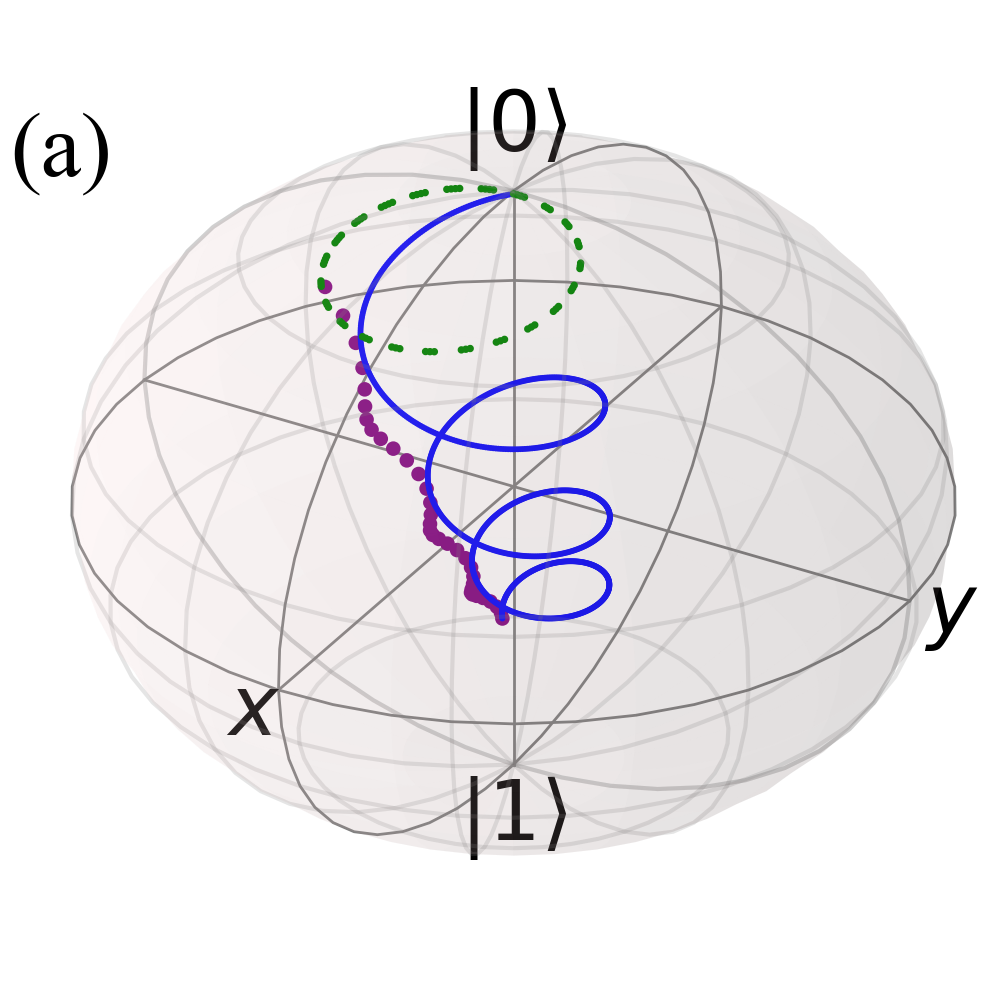}\qquad\qquad
\includegraphics[width=0.35\linewidth]{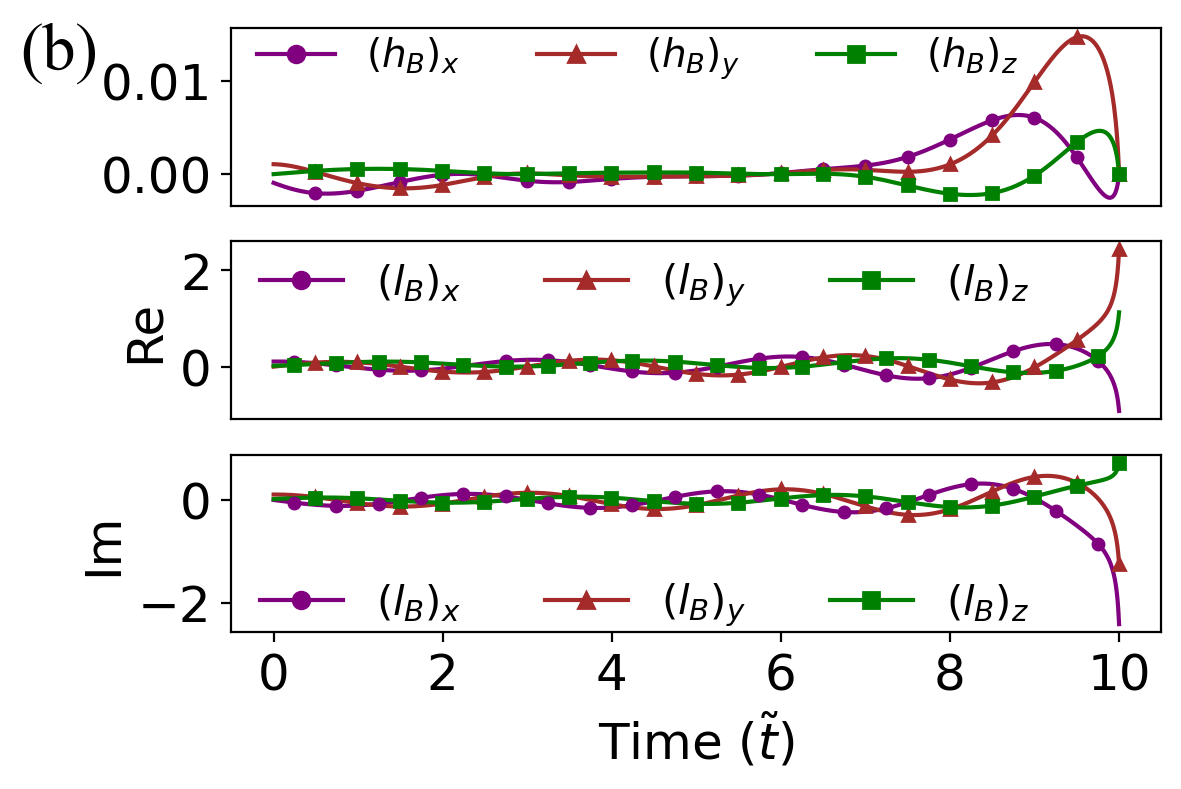}
\caption{Reversing noisy quantum dynamics while keeping the Hamiltonian dynamics. (a) The blue solid line and purple dotted line describe the quantum state's trajectory under the forward dynamics ${\cal L}$ and the reverse dynamics ${\cal L}'_B(\rho) = -i [H_B' , \rho] + {\cal D}[L_B'](\rho)$, respectively. After applying the reverse dynamics for time $\tau = 10$, the quantum state ends up to the unitary evolved state $U_\tau \ket{0}\bra{0} U_\tau^\dagger$, instead of being recovered to the initial state $\ket{0}\bra{0}$. The green dashed line shows the orbit given by the noise-free unitary operation $U_t \ket{0}\bra{0} U^\dagger_t$ for $t \in [0, \tau]$. (b) The recovery Hamiltonian $H_B' = \boldsymbol{h_B'} \cdot \boldsymbol{\sigma}$ and jump operator $L'_{B} = \boldsymbol{l_B'} \cdot \boldsymbol{\sigma}$ for the reverse dynamics.}
\label{fig:Suppl_no_rev_H}
\end{figure}

\subsection{B. Dissipation engineering of the two-level system}
In this session, we provide a more detailed discussion on engineering the jump operators of the two-level system. As explained in the main text, this can be done by introducing an ancillary system that interacts with the system in the following form of interaction Hamiltonian:
\begin{equation}
H_{\rm int}^{(\mu)}(\tilde t) = \frac{1}{2} \left( L_{B,\mu}^\dagger(\tilde t) \otimes \sigma_-^{(\mu)} + L_{B,\mu}(\tilde t) \otimes \sigma_+^{(\mu)} \right).
\end{equation}
If the dissipation of the ancillary system $\Gamma_a {\cal D} [\sigma_-^{(\mu)}]$ is strong the system's effective dynamics can be well approximated as the jump operator $\left( L_{B,\mu} /\sqrt{\Gamma_a}\right)$. Hence, the Hamiltonian of the total system becomes $H_{sa}(\tilde t) = H_B(\tilde t) +  \sqrt{\Gamma_a} \sum_\mu H_{\rm int}^{(\mu)}(\tilde t)$.
For the qubit system, $H_B$ can be fully described by $\sigma_{x,y,z}$, while $H_{\rm int}^{(\mu)}$ requires two-qubit Pauli operations $\sigma_{x,y,z} \otimes \sigma_{x,y}^{(\mu)}$ between the system and ancilla.

Another physically relevant situation is when the reverse dynamics contain uncontrollable parts. This commonly happens when the open quantum dynamics contains an uncontrollable dissipation $\Ddiss$. In this case, the physically achievable backward process would be ${\cal L}_B^{\rm phys.} =  \Ddiss + \Lengr$, where only $\Lengr$ can be engineered in a time-dependent manner. Nevertheless, the contribution from $\Ddiss$ can be suppressed by applying a strong controlled dynamics for short times. By rescaling $\tilde t \rightarrow \tilde t/\xi$, $H_B \rightarrow \xi H_B$, and $H_{\rm int}^{(\mu)} \rightarrow \sqrt{\xi} H_{\rm int}^{(\mu)}$, we effectively increase the strength of the jump operator $L_{B,\mu} \rightarrow \sqrt{\xi} L_{B,\mu}$, providing an effective dynamics as follows:
\begin{equation}
{\cal L}_B^\xi = \Ddiss + \xi {\cal L}_B.
\end{equation}

We provide an example by taking the forward dynamics $H  = 0.3 \sigma_x + \sigma_z$ and $L = 0.4 \sigma_-$, as studied in the main text. When assuming full controllability of the reverse dynamics ($\Ddiss = 0$), the recovery fidelity ${\cal F}(\gamma_{\tau - \tilde t}, \tilde \gamma_{\tilde t}) = {\rm Tr} \left( \tilde \gamma_{\tilde t} ^{1/2} \gamma_{\tau - \tilde t} \tilde \gamma_{\tilde t} ^{1/2} \right)^{1/2}$, where $\tilde \gamma_{\tilde t} = {\cal T} \left[ e^{\int_0^{{\tilde t}/\xi} {\cal L}_B^\xi d \tilde t'}  \right] (\gamma_\tau)$, becomes nearly $1$ for a sufficiently large $\Gamma$ as shown in Fig.~\ref{fig:Suppl_qubit_phys}. On the other hand, uncontrollable dissipation (${\cal D}_{\rm diss} \neq 0$) prevents the reverse dynamics to provide the perfect recovery. Nevertheless, by taking $\xi \geq 1$ one can suppress the effect of the uncontrollable dissipation to retrieve a high recovery fidelity as far as the adiabatic elimination condition $ \|\xi H_B\| , \| \sqrt{\xi} H_{\rm int} \| \ll \Gamma $ is valid (see Fig.~\ref{fig:Suppl_qubit_phys}).

\begin{figure}[ht]
\includegraphics[width=0.45\linewidth]{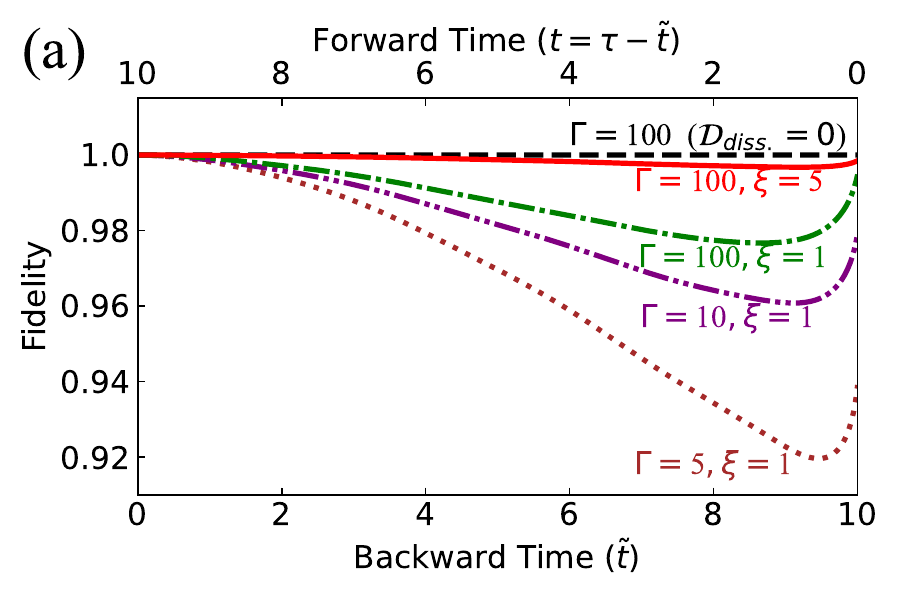}
\includegraphics[width=0.45\linewidth]{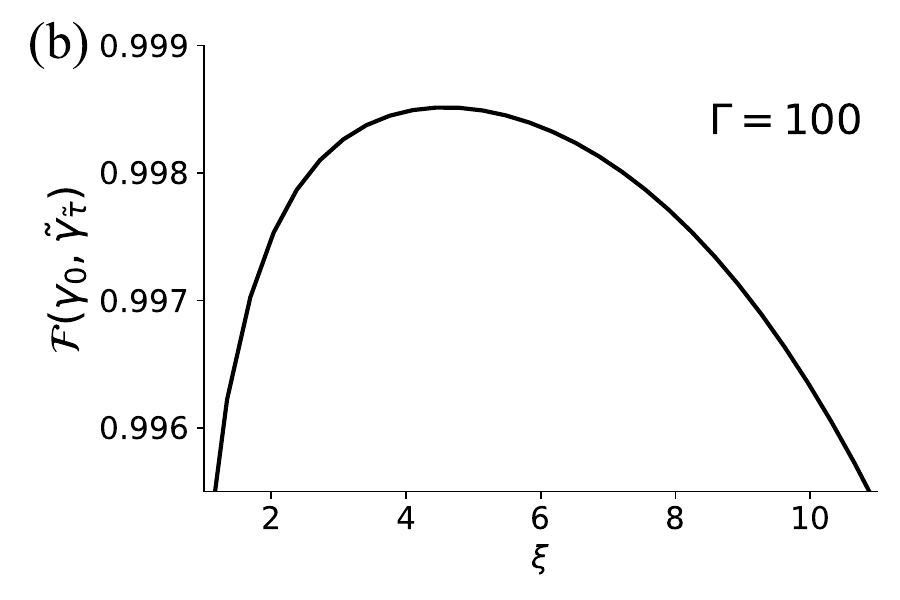}
\caption{(a) The recovery fidelity ${\cal F}(\gamma_{\tau - \tilde t}, \tilde \gamma_{\tilde t})$ at each time $\tilde t \in [0, \tau]$ evaluated for various $\Gamma$ and $\xi$. (b) Fidelity ${\cal F}(\gamma_0, \tilde \gamma_{\tilde \tau})$ between the initial state $\gamma_0$ and the recovered state $\tilde \gamma_{\tilde \tau}$ by changing $\xi$ for $\Gamma = 100$. The highest recovery fidelity is achieved for $\xi \approx 5$.}
\label{fig:Suppl_qubit_phys}
\end{figure}

\newpage
\section{III. Continuous recovery map with time-independent dissipation}
\subsection{A. Time-independent reverse dynamics for a fixed reference state}
In this section, we show that time-independent reverse dynamics can be constructed by considering a fixed state $\gamma$, instead of the trajectory $\gamma_t$. For a forward dynamics ${\cal L}$ described by Hamiltonian $H$ and jump operators $L_\mu$, the stationary dynamics of any full-rank state $\gamma$ can be constructed as
\begin{equation}
{\cal L}_S(\rho) = \left( {\cal L} + {\cal L}_B \right) (\rho),
\label{eq:Suppl_rev_static}
\end{equation}
where the time-independent recovery dynamics is defined as
\begin{equation}
{\cal L}_B(\rho)=  -i [H_B, \rho] + \sum_\mu {\cal D}[\gamma^{\frac{1}{2}} L_{\mu}^\dagger \gamma^{-\frac{1}{2}}] (\rho),
\label{eq:Supp_rererev}
\end{equation}
with $H_B = -H + \sum_\mu H_C(\gamma,L_\mu)$. Under the dynamics ${\cal L}_S$, $\gamma$ becomes a steady state satisfying
$$
\dot \gamma = {\cal L}_S (\gamma) = 0.
$$
\begin{proof}
This can be seen as a direct consequence of the property of the reverse dynamics ${\cal L}(\gamma) = -{\cal L}_B(\gamma)$. Nevertheless, we explicitly show that ${\cal L}_S(\gamma) = 0$ by noting that
$$
\begin{aligned}
(-i) \sum_\mu [H_C(\gamma,L_\mu),\gamma] 
&=  i \sum_\mu \sum_{\lambda, \lambda'} \left[(\lambda - \lambda') \left(\frac{-i}{2}\right) \left( \frac{\sqrt{\lambda}- \sqrt{\lambda'}}{\sqrt{\lambda} + \sqrt{\lambda'}} \right)\bra{\lambda} M_\mu(\gamma) \ket{ \lambda'}  \right] \ket{\lambda} \bra{\lambda'} \\
&= \frac{1}{2} \sum_\mu \sum_{\lambda, \lambda'} \left[ (\sqrt{\lambda} - \sqrt{\lambda'} )^2 \bra{\lambda} M_\mu(\gamma) \ket{ \lambda'}  \right] \ket{\lambda} \bra{\lambda'} \\
&= \frac{1}{2} \sum_\mu \left[ \gamma M_\mu(\gamma) + M_\mu(\gamma) - 2 \gamma^{\frac{1}{2}} M_\mu(\gamma) \gamma^{\frac{1}{2}} \right]\\
&= \frac{1}{2} \{ M_\mu (\gamma), \gamma \} - \gamma^{\frac{1}{2}} M_\mu(\gamma) \gamma^{\frac{1}{2}} \\
&= \frac{1}{2} \left \{ L_\mu^\dagger L_\mu, \gamma \right \} + \frac{1}{2} \left \{ \left( \gamma^{\frac{1}{2}} L_{\mu}^\dagger \gamma^{-\frac{1}{2}} \right)^\dagger \left( \gamma^{\frac{1}{2}} L_{\mu}^\dagger \gamma^{-\frac{1}{2}} \right), \gamma \right \}  - \gamma^{\frac{1}{2}} L_\mu^\dagger L_\mu \gamma^{\frac{1}{2}} - L_\mu \gamma L_\mu^\dagger \\
&= - {\cal D}[L_\mu](\gamma) - {\cal D}[\gamma^{\frac{1}{2}} L_\mu^\dagger \gamma^{-\frac{1}{2}}] (\gamma).
\end{aligned}
$$
Here, we used $H_C(\gamma, L_\mu) = - \frac{i}{2}\sum_{\lambda, \lambda'} \left( \frac{\sqrt{\lambda}- \sqrt{\lambda'}}{\sqrt{\lambda} + \sqrt{\lambda'}} \right)\bra{\lambda} M_\mu(\gamma) \ket{ \lambda'} \ket{\lambda} \bra{\lambda'}$ and $M_\mu(\gamma) = L_{\mu}^\dagger L_{\mu} + \gamma^{-\frac{1}{2}} L_{\mu} \gamma L_{\mu}^\dagger \gamma^{-\frac{1}{2}} =  L_{\mu}^\dagger L_{\mu} + \left( \gamma^{\frac{1}{2}} L_{\mu}^\dagger \gamma^{-\frac{1}{2}} \right)^\dagger \left( \gamma^{\frac{1}{2}} L_{\mu}^\dagger \gamma^{-\frac{1}{2}} \right)$.
Hence, we verify that
$$
{\cal L}_S(\gamma) = \left( {\cal L} + {\cal L}_B \right) (\gamma) = -i \sum_\mu [H_C(\gamma,L_\mu),\gamma] + {\cal D}[L_\mu](\gamma)  + {\cal D}[\gamma^{\frac{1}{2} }L_\mu^\dagger \gamma^{-\frac{1}{2}}](\gamma) = 0.
$$
\end{proof}

\subsection{B. Self-recovery property of the continuous recovery dynamics}
The self-recovery property of the recovery dynamics can be proven by showing that ${\cal L}_S$ satisfies the quantum following detailed balance relation \cite{Alhambra17} with respect to the fixed state $\gamma$:
\begin{equation}
\langle A, {\cal L}_S^\dagger  (B) \rangle_\gamma = \langle {\cal L}_S^\dagger(A),  B \rangle_\gamma,
\label{eq:Suppl_DB}
\end{equation}
for all operators $A$ and $B$. Here, the inner product is defined as 
$$
\langle A, B \rangle_\gamma := {\rm Tr} \left[ \gamma^{\frac{1}{2}} A^\dagger \gamma^{\frac{1}{2}} B\right].
$$
\begin{proof} 
By rearranging Eq.~\eqref{eq:Suppl_DB}, we obtain
$$
\begin{aligned}
\langle A, {\cal L}_S^\dagger  (B) \rangle_\gamma &= {\rm Tr} [\gamma^{\frac{1}{2}} A^\dagger \gamma^{\frac{1}{2}} {\cal L}_S^\dagger(B)] =  {\rm Tr} \left [ {\cal L}_S \left( \gamma^{\frac{1}{2}} A^\dagger \gamma^{\frac{1}{2}} \right) B \right ]\\
\langle {\cal L}_S^\dagger(A),  B \rangle_\gamma &= {\rm Tr} \left[ \gamma^{\frac{1}{2}} \left( {\cal L}_S^\dagger(A)\right)^\dagger \gamma^{\frac{1}{2}} B \right] = {\rm Tr} \left[ \gamma^{\frac{1}{2}} {\cal L}_S^\dagger(A^\dagger) \gamma^{\frac{1}{2}} B \right].
\end{aligned}
$$
Hence, we note that the detailed balance condition is equivalent to 
\begin{equation}
{\cal L}_S(\gamma^{\frac{1}{2}} \bullet \gamma^{\frac{1}{2}}) = \gamma^{\frac{1}{2}}  {\cal L}_S^\dagger (\bullet) \gamma^{\frac{1}{2}} \Longleftrightarrow \gamma^{-\frac{1}{2}} {\cal L}_S(\gamma^{\frac{1}{2}} \bullet \gamma^{\frac{1}{2}}) \gamma^{-\frac{1}{2}} = {\cal L}_S^\dagger (\bullet).
\label{eq:Suppl_db_equiv}
\end{equation}
The above equation can be shown as
$$
\begin{aligned}
&\gamma^{-\frac{1}{2}} {\cal L}_S(\gamma^{\frac{1}{2}} \bullet \gamma^{\frac{1}{2}}) \gamma^{-\frac{1}{2}} \\
&\quad= (-i) \sum_\mu  \gamma^{-\frac{1}{2}} H_C(\gamma,L_\mu)\gamma^{\frac{1}{2}} (\bullet)+ i  (\bullet)  \sum_\mu  \gamma^{\frac{1}{2}} H_C(\gamma,L_\mu)\gamma^{-\frac{1}{2}} \\
&\qquad + \sum_\mu \left( \gamma^{-\frac{1}{2}} L_\mu \gamma^{\frac{1}{2}} \bullet \gamma^{\frac{1}{2}} L_\mu^\dagger \gamma^{-\frac{1}{2}} -\frac{1}{2} \{ \gamma^{-\frac{1}{2}} L_\mu^\dagger L_\mu \gamma^{\frac{1}{2}}, \bullet \} \right) + \sum_\mu \left(\gamma^{-\frac{1}{2}} L_\mu \gamma^{\frac{1}{2}} \bullet \gamma^{\frac{1}{2}} L_\mu^\dagger \gamma^{-\frac{1}{2}} -\frac{1}{2} \{ \gamma^{-1} L_\mu \gamma L_\mu^\dagger , \bullet \} \right) \\
&\quad= \sum_\mu \left( {\cal D}^\dagger [L_\mu] + {\cal D}^\dagger[\gamma^{\frac{1}{2}} L_\mu^\dagger \gamma^{-\frac{1}{2}}] \right) (\bullet)\\
& \qquad + \sum_\mu \left( (-i) \gamma^{-\frac{1}{2}} H_C(\gamma, L_\mu) \gamma^{\frac{1}{2}} -\frac{1}{2} \gamma^{-\frac{1}{2}} L_\mu^\dagger L_\mu \gamma^{\frac{1}{2}} - \frac{1}{2} \gamma^{-1} L_\mu \gamma L_\mu^\dagger  + \frac{1}{2} L_\mu^\dagger L_\mu + \frac{1}{2} \gamma^{-\frac{1}{2}}L_\mu \gamma L_\mu^\dagger \gamma^{-\frac{1}{2}} \right) (\bullet) \\
& \qquad +  (\bullet) \sum_\mu  \left( i \gamma^{\frac{1}{2}} H_C(\gamma, L_\mu) \gamma^{-\frac{1}{2}} -\frac{1}{2} \gamma^{-\frac{1}{2}} L_\mu^\dagger L_\mu \gamma^{\frac{1}{2}} - \frac{1}{2} \gamma^{-1} L_\mu \gamma L_\mu^\dagger  + \frac{1}{2} L_\mu^\dagger L_\mu + \frac{1}{2} \gamma^{-\frac{1}{2}}L_\mu \gamma L_\mu^\dagger \gamma^{-\frac{1}{2}} \right)\\
&\quad =  \sum_\mu \left( {\cal D}^\dagger [L_\mu] + {\cal D}^\dagger[\gamma^{\frac{1}{2}} L_\mu^\dagger \gamma^{-\frac{1}{2}}] \right) (\bullet) + i \sum_\mu [H_C(\gamma,L_\mu), \bullet]\\
&\quad = {\cal L}_S^\dagger (\bullet),
\end{aligned}
$$
by noting that
$$
(-i) \gamma^{-\frac{1}{2}} H_C(\gamma, L_\mu) \gamma^{\frac{1}{2}} -\frac{1}{2} \gamma^{-\frac{1}{2}} L_\mu^\dagger L_\mu \gamma^{\frac{1}{2}} - \frac{1}{2} \gamma^{-1} L_\mu \gamma L_\mu^\dagger = \gamma^{-\frac{1}{2}} \left[ -i H_C(\gamma,L_\mu) - \frac{1}{2} M_\mu(\gamma)  \right] \gamma^{\frac{1}{2}} = i H_C(\gamma, L_\mu) - \frac{1}{2} M_\mu(\gamma).
$$

The detailed balance relation in Eq.~\eqref{eq:Suppl_DB} implies that the Petz recovery channel of ${\cal N}_{\tau} = {\cal T}\left[ \exp \left( \int_0^{\tau} {\cal L}_S dt \right) \right]$ is nothing but itself. This can be shown as
$$
{\cal R}_{\gamma, {\cal N}_\tau} = \JJ{\gamma}{\frac{1}{2}} \circ {\cal N}^\dagger \circ \JJ{ \gamma}{-\frac{1}{2}}
= {\cal T}\left[ \exp \left( \int_0^{\tau} \JJ{\gamma}{\frac{1}{2}} \circ {\cal L}_S^\dagger \circ \JJ{ \gamma}{-\frac{1}{2}} dt \right) \right]= {\cal T}\left[ \exp \left( \int_0^{\tau}  {\cal L}_S dt \right) \right],
$$
for any time $\tau$. This can be shown by the fact that Eq.~\eqref{eq:Suppl_db_equiv} is equivalent to 
$$
\gamma^{-\frac{1}{2}} {\cal L}_S(\gamma^{\frac{1}{2}} \bullet \gamma^{\frac{1}{2}}) \gamma^{-\frac{1}{2}} = {\cal L}_S^\dagger (\bullet) \Longleftrightarrow \left( \JJ{\gamma}{\frac{1}{2}} \circ {\cal L}_S^\dagger \circ \JJ{ \gamma}{-\frac{1}{2}} \right) (\blacksquare) = \gamma^{\frac{1}{2}} {\cal L}_S^\dagger (\gamma^{-\frac{1}{2}} \blacksquare \gamma^{-\frac{1}{2}}) \gamma^{\frac{1}{2}} = {\cal L}_S (\blacksquare),
$$
by taking $ \bullet = \gamma^{-\frac{1}{2}} \blacksquare \gamma^{-\frac{1}{2}}$ and multiplying the both sides with $\gamma^{\frac{1}{2}}$.
\end{proof}
It is worth noting that ${\cal L}_S$ is not necessarily a Davies map \cite{Davies74}, but enjoys the self-recovery property.

\subsection{C. Continuous recovery dynamics for a stabilizer code}
We provide a canonical form of the continuous recovery map of a $n$-qubit stabilizer code, whose dissipation is given by
$$
{\cal L} = \sum_\mu \Gamma_\mu {\cal D}[E_\mu],
$$
where $E_{\mu} \in \langle \sigma_x, \sigma_y, \sigma_z \rangle^n$ in the $n$-qubit Pauli group. For a given stabilizer ${\cal S}$, the code space ${\cal C}$ can be characterized as a set of quantum states $\ket\psi$ such that $ S_i \ket\psi =  \ket\psi$ $\forall S_i \in {\cal S}$. We then take $Q = - \sum_{S_i \in \bar{\cal S}}  S_i$ for a subset of the stabilizer $\bar{\cal S} \subset {\cal S}$ and a fixed state $\gamma = e^{-\beta Q} /{\rm Tr}[e^{-\beta Q}]$ with $\beta \geq 0$. This leads to the following recovery dynamics:
\begin{equation}
{\cal L}_B = \sum_\mu \Gamma_\mu {\cal D}\left[E_{\mu} \prod_{S_i \in \bar{\cal S}_\mu} \left[ (\cosh \beta) \mathbb{1} - (\sinh{\beta}) S_i \right] \right],
\label{eq:Suppl_Stab}
\end{equation}
where $\bar{\cal S}_\mu = \{ S_i \in \bar{\cal S} | \{S_i, E_\mu \} = 0 \}$.

\begin{proof}
We first note that every $E_\mu$ is a Hermitian operator such that $E_\mu = E_\mu^\dagger$ as being an element of the Pauli group. We also note that $$
e^{-\beta Q} = \prod_{S_i \in \bar{\cal S}} e^{\beta S_i} =  \prod_{S_i \in \bar{\cal S}}  \left[ \cosh(\beta) \mathbb{1} + \sinh(\beta) S_i \right],
$$
from the facts that $[S_i, S_j] = 0$ and $S_i^2 = \mathbb{1}$ for all $ S_i, S_j \in {\cal S}$.  As every element of the stabilizer either commutes $[S_i, E_\mu] = 0 $ or anti-commutes $\{S_i, E_\mu \} = 0$ for a given $E_\mu$ in a Pauli group, we define a set of elements of the stabilizer that anti-commuites with $E_\mu$ to be $\bar{\cal S}_\mu = \{ S_i \in \bar{\cal S} | \{S_i, E_\mu \} = 0 \}$. We then show that the jump operator of the recovery dynamics can be written as
$$
\begin{aligned}
\gamma^{\frac{1}{2}} E_\mu \gamma^{-\frac{1}{2}} 
&= e^{-\beta Q /2} E_{\mu} e^{\beta Q/2} \\
&= \left( \prod_{S_i \in \bar{\cal S}} e^{\beta S_i/2} \right) E_{\mu} \left( \prod_{S_i \in \bar{\cal S}} e^{-\beta S_i/2} \right)\\
&= E_{\mu} \left( \prod_{S_i \in \bar{\cal S}_\mu} e^{-\beta S_i/2} \right) \left( \prod_{S_i \in \bar{\cal S} - \bar{\cal S}_\mu} e^{\beta S_i/2} \right)  \left( \prod_{S_i \in \bar{\cal S} - \bar{\cal S}_\mu} e^{-\beta S_i/2} \right) \left( \prod_{S_i \in \bar{\cal S}_\mu} e^{-\beta S_i/2} \right)\\
&= E_{\mu} \left( \prod_{S_i \in \bar{\cal S}_\mu} e^{-\beta S_i} \right)\\
&= E_{\mu} \left( \prod_{S_i \in \bar{\cal S}_\mu} \left[ \cosh(\beta) \mathbb{1} - \sinh(\beta) S_i \right] \right).
\end{aligned}
$$
We then show that the recovery hamiltonian becomes zero, i.e.,
$$
H_C(\gamma,E_\mu) = 0.
$$
In order to prove this, we evaluate
$$
M_\mu(\gamma) = E_\mu^\dagger E_\mu + e^{\beta Q/2} E_\mu e^{-\beta Q} E_\mu^\dagger e^{\beta Q/2}= \mathbb{1} +  \left( \prod_{S_i \in \bar{\cal S}_\mu} e^{-\beta S_i} \right) E_{\mu} E_\mu^\dagger \left( \prod_{S_i \in \bar{\cal S}_\mu} e^{-\beta S_i} \right)= \mathbb{1} +  \left( \prod_{S_i \in \bar{\cal S}_\mu} e^{-2\beta S_i} \right),
$$
which implies that $[Q, M_\mu(\gamma)] = 0$. Hence, we conclude that the reverse Hamiltonian
$$
H_C(\gamma, E_\mu) = - \frac{i}{2}\sum_{\lambda, \lambda'} \left( \frac{\sqrt{\lambda}- \sqrt{\lambda'}}{\sqrt{\lambda} + \sqrt{\lambda'}} \right)\bra{\lambda} M_\mu(\gamma) \ket{ \lambda'} \ket{\lambda} \bra{\lambda'},
$$
becomes zero as every $\ket{\lambda}\bra{\lambda}$, an eigenstate of $\gamma \propto e^{-\beta Q}$, commutes with $M_\mu(\gamma)$ for all $E_\mu$.
\end{proof}

We note that in a large limit of $\beta$, the jump operators of the recovery dynamics becomes
$$
(\cosh \beta) \mathbb{1} - (\sinh{\beta}) S_i \approx e^{\beta} \left( \frac{\mathbb{1} - S_i }{2} \right),
$$
where $ \frac{\mathbb{1} - S_i }{2}$ is a projection onto the Hilbert space having a non-trivial syndrome detected for a stabilizer element $S_i$. Hence, in this limit, Eq.~\eqref{eq:Suppl_Stab} becomes 
$$
{\cal L}_B \approx \sum_\mu \Gamma_\mu e^{2 |\bar{\cal S}_\mu| \beta} {\cal D} \left[E_\mu \prod_{S_i \in \bar{\cal S}_\mu} \left( \frac{ \mathbb{1} -  S_i}{2}\right) \right],
$$
which can be interpreted as a continuous quantum error correction protocol, where the syndrome detection and correction are operating continuously in time.

\subsection{D. Continuous recovery dynamics for a general code space}
For any given code space ${\cal C}$, one can construct a reverse dynamics by taking
$$
\gamma = \frac{e^{-\beta Q}}{{\rm Tr} \left[ e^{-\beta Q} \right]},
$$
where $Q = -P_{\cal C}$ with $P_{\cal C} = \sum_{\ket\psi \in {\cal C}} \ket\psi \bra\psi$ being a projection onto the code space. By noting that $e^{-\beta Q} = e^{\beta P_{\cal C}} = e^\beta P_{\cal C} + (\mathbb{1} - P_{\cal C})$, we can express the jump operators $ L_{B,\mu} = \gamma^{\frac{1}{2}} L_\mu^\dagger \gamma^{-\frac{1}{2}}$ of the reverse dynamics in Eq.~\eqref{eq:Supp_rererev} by noting that
\begin{equation}
\begin{aligned}
\gamma^{\frac{1}{2}} L_\mu^\dagger \gamma^{-\frac{1}{2}} &= e^{(\beta/2) P_{\cal C}} L_\mu^\dagger e^{-(\beta/2) P_{\cal C}} \\
&= \left( e^{\beta/2} P_{\cal C} + (\mathbb{1} - P_{\cal C}) \right) L_\mu^\dagger \left( e^{-\beta/2} P_{\cal C} + (\mathbb{1} - P_{\cal C}) \right) \\
&= e^{\beta/2} \left( P_{\cal C} +  e^{-\beta/2} (\mathbb{1} - P_{\cal C}) \right) L_\mu^\dagger \left( e^{-\beta/2} P_{\cal C} + (\mathbb{1} - P_{\cal C}) \right).
\end{aligned}
\label{eq:Suppl_P_rev_L}
\end{equation}
Meanwhile, the reverse Hamiltonian can be expressed as
$$
\begin{aligned}
H_C(\gamma, L_\mu) &= \left(-\frac{i}{2} \right) \sum_{\lambda, \lambda'} \left( \frac{\sqrt{\lambda}- \sqrt{\lambda'}}{\sqrt{\lambda} + \sqrt{\lambda'}} \right)\bra{\lambda} M_\mu(\gamma) \ket{ \lambda'} \ket{\lambda} \bra{\lambda'} \\
&=  \left(-\frac{i}{2} \right) \left( \sum_{\ket\lambda \in {\cal C}, \ket{\lambda'} \notin {\cal C}} \left( \frac{e^{\beta/2} - 1}{e^{\beta/2} + 1} \right) \bra{\lambda} M_\mu(\gamma) \ket{ \lambda'} \ket{\lambda} \bra{\lambda'} + \sum_{\ket\lambda \in {\cal C}, \ket{\lambda'} \notin {\cal C}} \left( \frac{ 1- e^{\beta/2} }{1 + e^{\beta/2} } \right) \bra{\lambda} M_\mu(\gamma) \ket{ \lambda'} \ket{\lambda} \bra{\lambda'} \right) \\
&= \left(-\frac{i}{2} \right) \left( \frac{ 1- e^{-\beta/2} }{1 + e^{-\beta/2} } \right) \left( P_{\cal C} M_\mu(\gamma) (\mathbb{1} - P_{\cal C}) -  (\mathbb{1} - P_{\cal C}) M_\mu(\gamma) P_{\cal C} \right)\\
&=\left(-\frac{i}{2} \right) \left( \frac{ 1- e^{-\beta/2} }{1 + e^{-\beta/2} } \right) \left[P_{\cal C}, M_\mu(\gamma) \right ].
\end{aligned}
$$
By substituting
$$
\begin{aligned}
M_\mu(\gamma) &= L_\mu^\dagger L_\mu + \gamma^{-\frac{1}{2}} L_\mu \gamma L_\mu^\dagger \gamma^{-\frac{1}{2}} \\
&=L_\mu^\dagger L_\mu + \left( e^{-\beta/2} P_{\cal C} + (\mathbb{1} - P_{\cal C}) \right) L_\mu  \left( e^{\beta} P_{\cal C} + (\mathbb{1} - P_{\cal C}) \right) L_\mu^\dagger  \left( e^{-\beta/2} P_{\cal C} + (\mathbb{1} - P_{\cal C}) \right),
\end{aligned}
$$
to $H_C$, we have
\begin{equation}
H_C(\gamma, L_\mu) = \left(-\frac{i}{2} \right) \left( \frac{ 1- e^{-\beta/2} }{1 + e^{-\beta/2} } \right) \left[ P_{\cal C}, L_\mu^\dagger L_\mu + e^{\beta/2} L_\mu (P_{\cal C} + e^{-\beta} (\mathbb{1} - P_{\cal C})) L_\mu^\dagger \right].
\label{eq:Suppl_P_rev_H}
\end{equation}
In the limit of large $\beta$, Eqs.~\eqref{eq:Suppl_P_rev_L} and \eqref{eq:Suppl_P_rev_H} lead to
$$
{\cal D} [\gamma^{\frac{1}{2}} L_\mu^\dagger \gamma^{-\frac{1}{2}}] \approx e^\beta {\cal D} \left[ P_{\cal C} L_\mu^\dagger \left( \mathbb{1} - P_{\cal C} \right)\right]
$$
and
$$
H_C(\gamma,L_\mu) \approx \left( - \frac{i}{2} \right) e^{\beta/2} \left( P_{\cal C} L_\mu P_{\cal C} L_\mu^\dagger - L_\mu P_{\cal C} L_\mu^\dagger P_{\cal C} \right).
$$
Since the contribution of the jump operator becomes dominant for $\beta \gg 1$, the effective recovery dynamics can be written as
$$
{\cal L}_B \approx e^{\beta} \sum_\mu {\cal D}[ P_{\cal C} L_\mu^\dagger \left( \mathbb{1} - P_{\cal C} \right)].
$$


\section{IV. Other noise models and code space constructions}

\subsection{A. Bit and phase flipping without correlated noise}
We first study a noise model where bit and phase flipping errors continuously happen to each physical qubit, whose master equation is expressed as
$$
{\cal L} = \Gamma \sum_{i=1}^n {\cal D}[\sigma_x^{(i)}] + \Gamma \sum_{i=1}^n {\cal D}[\sigma_z^{(i)}].
$$
This noise model equivalent to independent noise channels acting on each physical qubits, i.e., $e^{{\cal L} t} (\rho) =  ({\cal E}^{(n)} \circ \cdots \circ {\cal E}^{(1)}) (\rho)$, where ${\cal E}^{(i)}(\rho) =  (1-p)^2 \rho + p(1-p) (\sigma_x^{(i)} \rho \sigma_x^{(i)} +\sigma_z^{(i)} \rho \sigma_z^{(i)}) + p^2 \sigma_y^{(i)}\rho \sigma_y^{(i)} $ with $p = (1 - e^{-2\Gamma t})/2$. 

We construct the recovery dynamics based on the $[\![5,1,3]\!]$ code as in the main text, by taking $\gamma = e^{-\beta Q}/{\rm Tr} [e^{-\beta Q}]$ with $Q = -\sum_{i=1}^5 \sigma_x^{(i)} \sigma_z^{(i+1)} \sigma_z^{(i+2)} \sigma_x^{(i+3)}$ and by defining $\sigma_{x,z}^{(5l + i)} = \sigma_{x,z}^{(i)}$ for $l \in \mathbb{Z}$.
From Eq.~\eqref{eq:Suppl_Stab}, we obtain
$$
\begin{aligned}
{\cal L}_B &= \Gamma \sum_{i=1}^n {\cal D}\left[ e^{-\beta Q/2} \sigma_x^{(i)} e^{\beta Q/2} \right] + \Gamma \sum_{i=1}^n {\cal D}\left[ e^{-\beta Q/2} \sigma_z^{(i)} e^{\beta Q/2} \right] \\
&= \Gamma \sum_{i=1}^n {\cal D}\left[  \sigma_x^{(i)} \left( \cosh(\beta) \mathbb{1} - \sinh(\beta) \sigma_x^{(i-1)}\sigma_z^{(i)} \sigma_z^{(i+1)}\sigma_x^{(i+2)} \right) \left( \cosh(\beta) \mathbb{1} - \sinh(\beta) \sigma_x^{(i-2)}\sigma_z^{(i-1)} \sigma_z^{(i)}\sigma_x^{(i+1)} \right)  \right] \\
&\quad+  \Gamma \sum_{i=1}^n {\cal D}\left[  \sigma_z^{(i)} \left( \cosh(\beta) \mathbb{1} - \sinh(\beta) \sigma_x^{(i)}\sigma_z^{(i+1)} \sigma_z^{(i+2)}\sigma_x^{(i+3)} \right) \left( \cosh(\beta) \mathbb{1} - \sinh(\beta) \sigma_x^{(i-3)}\sigma_z^{(i-2)} \sigma_z^{(i-1)}\sigma_x^{(i)} \right)  \right].
\end{aligned}
$$
Figure~\ref{fig:Suppl_5q_XZ} shows that the recovery dynamics becomes effective for most values of $p$ when $\beta \gtrsim 2$. We also note that the lower noise level $p$, the smaller $\beta$ is required to have an effective recovery, similar to the case in the main text, which contains additional $ZZ$ correlated noise terms. Even in the regime where QEC is not effective $(p \gtrsim 7\%)$, applying the recovery dynamics still can reduced the noise level.

\begin{figure}[ht]
\includegraphics[width=0.45\linewidth]{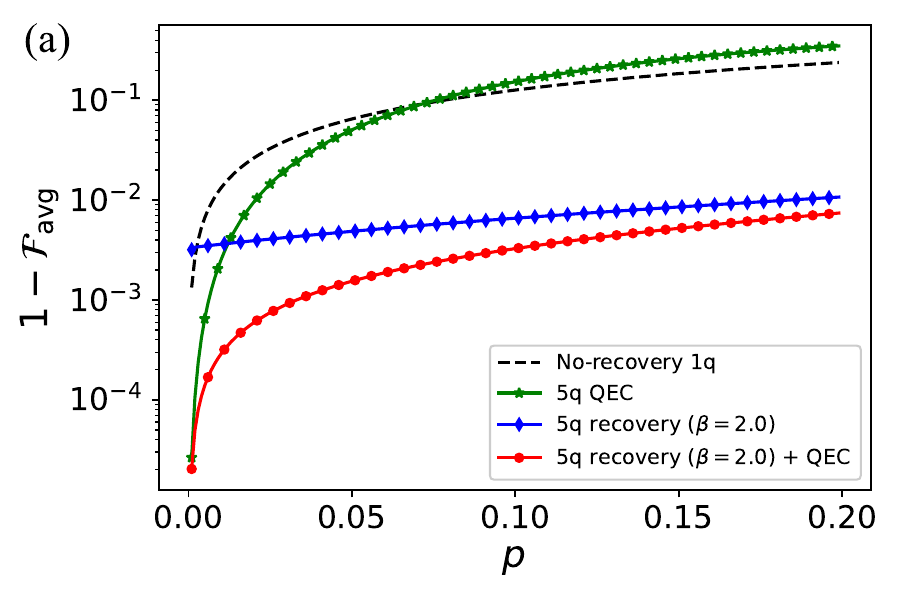}
\includegraphics[width=0.45\linewidth]{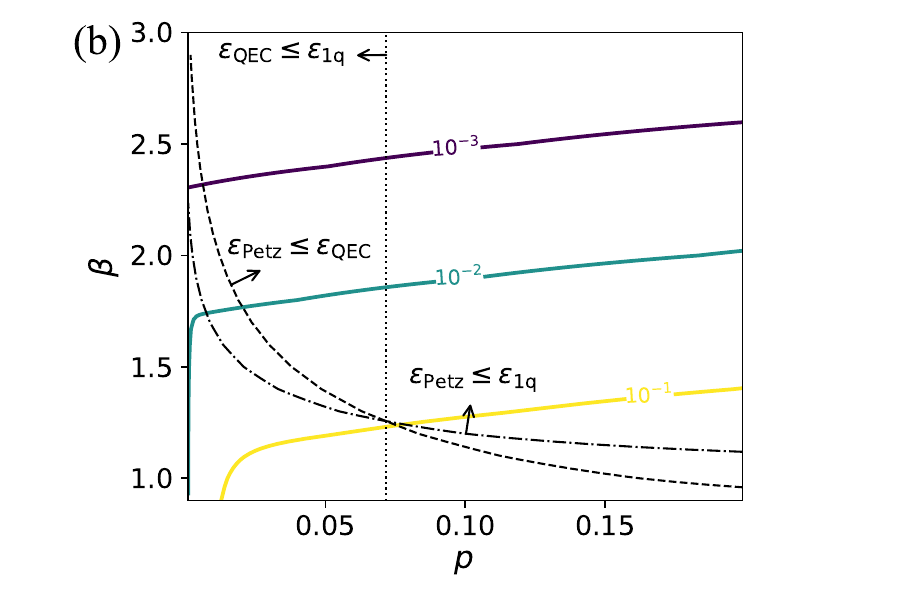}
\caption{(a) Comparison between the average infidelity $1 - {\cal F}_{\rm avg}$ of the $[\![5,1,3]\!]$ code space by applying QEC and the continuous recovery dynamics against bit and phase flipping noise.
(b) Average infidelity $1 - {\cal F}_{\rm avg}$ for applying the continuous recovery dynamics.  Dotted, dot-dashed, dashed lines describe the region where the QEC is effective, the continuous Petz recovery is effective, and the continuous Petz recovery is better than performing only QEC, respectively. Here, $\varepsilon_{\rm 1q}$, $\varepsilon_{\rm QEC}$, and $\varepsilon_{\rm Petz}$ refer to the average infidelities $1 - {\cal F}_{\rm avg}$ for a single qubit without encoding, the $[\![5,1,3]\!]$ code with QEC, and applying recovery dynamics to the $[\![5,1,3]\!]$ code without QEC, respectively.}
\label{fig:Suppl_5q_XZ}
\end{figure}

\subsection{B. Amplitude damping noise}
Next, we consider amplitude damping noise, described by the following Lindblad operator:
$$
{\cal L} = \Gamma_- \sum_{i=1}^n {\cal D}[\sigma_-^{(i)}].
$$
Similarly to the previous case, the noise model can be described as $e^{{\cal L} t} (\rho) =  ({\cal E}^{(n)} \circ \cdots \circ {\cal E}^{(1)}) (\rho)$, where
$$
{\cal E}^{(i)}(\rho) = E_0^{(i)} \rho E_0^{(i) \dagger} + E_1^{(i)} \rho E_1^{(i)\dagger}.
$$
Here, the Kraus operators of the individual amplitude damping channel for each $i$th qubit are
$$
\begin{aligned}
E_0^{(i)} &= \ket{0}_i \bra{0} + \sqrt{1-p_-} \ket{1}_i \bra{1}\\
E_1^{(i)} &= \sqrt{p_-} \ket{0}_i\bra{1},
\end{aligned}
$$
with $p_- = 1 - e^{-\Gamma_- t}$.

We adopt the $[\![4,1]\!]$ code that has been first studied in Ref.~\cite{Leung97} to construct approximate QEC against amplitude damping. The code space ${\cal C}$ is constructed by encoding a single logical qubit into four physical qubits as 
$$
\begin{aligned}
\ket{0}_L &= \frac{1}{\sqrt{2}} \left( \ket{0000} + \ket{1111} \right)\\
\ket{1}_L &= \frac{1}{\sqrt{2}} \left( \ket{0011} + \ket{1100} \right).
\end{aligned}
$$
We then take $Q = -P_{\cal C} = -(\ket{0}_L\bra{0} + \ket{1}_L\bra{1})$ to construct the recovery dynamics, by following Eqs.~\eqref{eq:Suppl_P_rev_L} and \eqref{eq:Suppl_P_rev_H}. The behaviors of the average fidelity by changing noise level $p_-$ and recovery strength $\beta$ is described in Fig.~\ref{fig:Suppl_4q_amp_damp}. One may note that $\beta$ to obtain effective recovery dynamics seems to be larger than the $[\![5,1,3]\!]$ code case, but we point out that the overall strength of the recovery dynamics for the two cases are comparable when comparing the jump operator's norms of the reverse dynamics.

\begin{figure}[ht]
\includegraphics[width=0.45\linewidth]{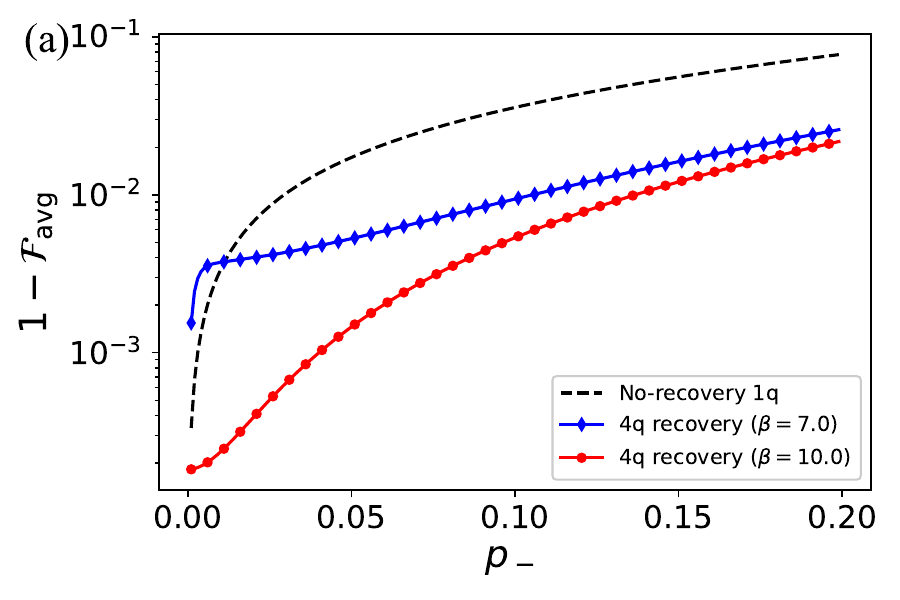}
\includegraphics[width=0.45\linewidth]{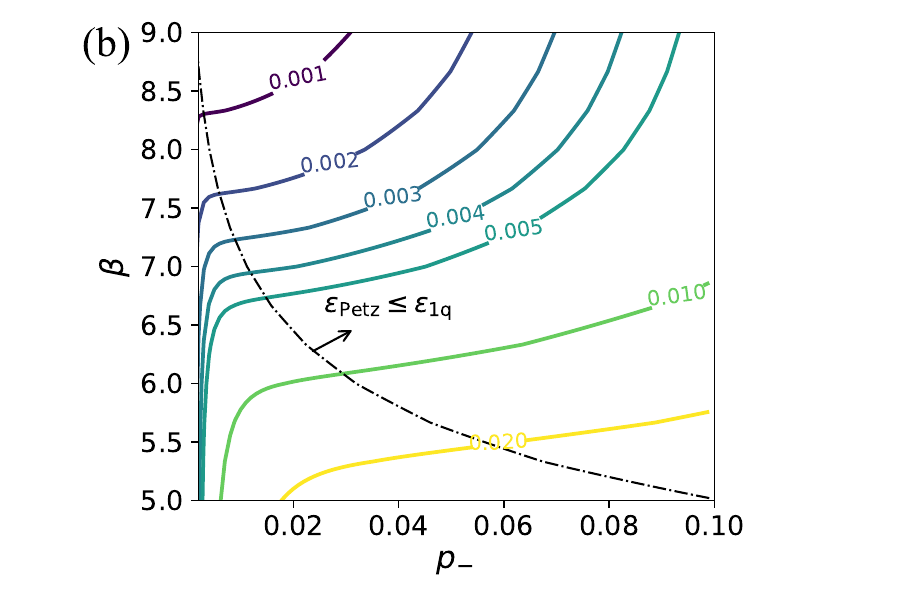}
\caption{(a) Comparison between the average infidelity $1 - {\cal F}_{\rm avg}$ of the $[\![4,1]\!]$ code by applying QEC and the continuous recovery dynamics against amplitude damping noise. (b) Average infidelity $1 - {\cal F}_{\rm avg}$ for applying the continuous recovery dynamics.  Dot-dashed line describes the region where the continuous Petz recovery is effective. Here, $\varepsilon_{\rm 1q}$ and $\varepsilon_{\rm Petz}$ refer to the average infidelities $1 - {\cal F}_{\rm avg}$ for a single qubit without encoding and applying recovery dynamics to the $[\![4,1]\!]$ code, respectively.}
\label{fig:Suppl_4q_amp_damp}
\end{figure}

\subsection{C. Optimizing the code space}
Suppose that a quantum state in a $d$-dimensional Hilbert space is encoded into an orthonormal basis ${\cal C} = \{ \ket{0}_L, \cdots, \ket{d-1}_L \}$.  The average fidelity of the code space after some time $\tau$ under the Lindblad dynamics ${\cal L}_S$ is given as
$
{\cal F}_{\rm avg} = \int_{\cal C} d\psi \bra \psi {\cal T} \left[ e^{ \int_0^\tau {\cal L}_S dt} \right] (\ket{\psi}\bra{\psi}) \ket{\psi},
$
where the integral is performed over the uniform Haar measure $d\psi$ of the code space ${\cal C}$. Instead of averaging over $d\psi$ in the code space, we use the entanglement fidelity which is closely related to the average fidelity \cite{Horodecki99, Nielsen02},
$$
{\cal F}_e = \bra{\Psi} \left( {\cal T} \left[ e^{ \int_0^{\delta t} {\cal L}_S dt} \right]  \otimes {\cal I} \right) (\ket\Psi\bra\Psi) \ket{\Psi}.
$$
One may consider optimizing the code space ${\cal C}$ to get the highest entanglement fidelity. However, such an optimization problem is a complex problem as ${\cal L}_S = {\cal L} + {\cal L}_B$ depends on the choice of the code space ${\cal C}$.

In order to detour this problem, we recall that the recovery dynamics can be considered as the Petz recovery map applied at short time $\delta t$ after the forward dynamics, i.e.,
$$
{\cal T} \left[ e^{ \int_0^{\delta t} {\cal L}_S dt} \right] \approx  e^{\delta t {\cal L}_B} e^{\delta t {\cal L}},
$$
where
$$
e^{{\cal L}_B \delta t} (\bullet) = \gamma^{\frac{1}{2}} e^{{\cal L}^\dagger \delta t} \left( \gamma_{\delta t}^{-\frac{1}{2}} (\bullet) \gamma_{\delta t}^{-\frac{1}{2}} \right) \gamma^{\frac{1}{2}},
$$
with $\gamma_{\delta t} = e^{{\cal L} \delta t}(\gamma)$. We also take $\gamma = e^{-\beta Q}/{\rm Tr}[e^{-\beta Q}]$ in the limit $\beta \rightarrow \infty$, so that $\gamma \approx P_{\cal C}/d$, where $d$ is a dimension of the code space. Under these assumptions, the fidelity of a quantum state in the code space $\ket\psi \in {\cal C}$ becomes
\begin{equation}
\begin{aligned}
\bra{\psi} {\cal T} \left[ e^{ \int_0^{\delta t} {\cal L}_S dt} \right]  (\ket\psi\bra\psi) \ket{\psi}
&\approx  \bra{\psi} \gamma^{\frac{1}{2}} e^{\delta t {\cal L}^\dagger} \left( \gamma_{\delta t}^{-\frac{1}{2}} e^{\delta t {\cal L}}  (\ket\psi\bra\psi) \gamma_{\delta t}^{-\frac{1}{2}}  \right) \gamma^{\frac{1}{2}} \ket{\psi}\\
&=\frac{1}{d} \bra{\psi} e^{\delta t {\cal L}^\dagger} \left( \gamma_{\delta t}^{-\frac{1}{2}} e^{\delta t {\cal L}}  (\ket\psi\bra\psi) \gamma_{\delta t}^{-\frac{1}{2}}  \right) \ket{\psi}\\
&=\frac{1}{d} {\rm Tr} \left[  e^{\delta t {\cal L}}( \ket{\psi} \bra{\psi}) \gamma_{\delta t}^{-\frac{1}{2}} e^{\delta t {\cal L}}  (\ket\psi\bra\psi) \gamma_{\delta t}^{-\frac{1}{2}}  \right].
\end{aligned}
\label{eq:Suppl_fid_formula}
\end{equation}

By following Eq.~\eqref{eq:Suppl_fid_formula}, we obtain an alternative form of the entanglement fidelity,
$$
{\cal F}_e =\frac{1}{d^3} \sum_{i,k=0}^{d-1} {\rm Tr} \left [ O_{ki}^{\delta t} \gamma_{\delta t}^{-\frac{1}{2}}  O_{ik}^{\delta t} \gamma_{\delta t}^{-\frac{1}{2}} \right] = \frac{1}{d^2} \sum_{i,k=0}^{d-1} {\rm Tr} \left [ O_{ki}^{\delta t} (P_{\cal C})_{\delta t}^{-\frac{1}{2}}  O_{ik}^{\delta t} (P_{\cal C})_{\delta t}^{-\frac{1}{2}} \right],
$$
where $O_{ki}^{\delta t} = e^{\delta t {\cal L}} (\ket{k}_L\bra{i})$ and $\gamma_{\delta t}^{-\frac{1}{2}} = (P_{\cal C}/d)_{\delta t}^{-\frac{1}{2}} = d^{\frac{1}{2}} (P_{\cal C})_{\delta t}^{-\frac{1}{2}} $. Note that only the forward dynamics ${\cal L}$ is needed to obtain the entanglement fidelity for a given code space. We then express the code space by a unitary operator as 
$$
{\cal C}_U = \{ U \ket{0}_L, \cdots, U\ket{d-1}_L \},
$$
where $\{ \ket{0}_L, \cdots, \ket{d-1}_L \}$ can be any initial guess of the code space. We also note that the unitary operator can be parametrised with 
$2^{2n}$ independent real parameters when considering $n$-physical qubits,. Finally, the optimization problem can be written as
\begin{equation}
\begin{aligned}
{\rm maximize:} &\qquad {\cal F}_e(U) = \frac{1}{d^2} \sum_{i,k=0}^{d-1} {\rm Tr} \left [ O_{ki}^{\delta t}(U) (P_{{\cal C}_U})_{\delta t}^{-\frac{1}{2}} O_{ik}^{\delta t}(U) (P_{{\cal C}_U})_{\delta t}^{-\frac{1}{2}} \right]\\
{\rm constraint:} &\qquad U^\dagger U = \mathbb{1},
\end{aligned}
\label{eq:Suppl_opt_prob}
\end{equation}
where $O_{ki}^{\delta t}(U) = e^{\delta t {\cal L}} (U \ket{k}\bra{i} U^\dagger)$ and $(P_{{\cal C}_U})_{\delta t}^{-\frac{1}{2}} = \left( e^{\delta t {\cal L}} (P_{\cal C})\right)^{-1/2}$.

We note that, however, ${\cal F}_e$ behaves singularly for $\delta t \rightarrow 0$ as $(P_{{\cal C}_U})_{\delta t}^{-\frac{1}{2}}$ becomes divergent. Hence, we take a finite $\delta t = 0.001$ to solve the optimization problem for this work. For the numerical optimization, we adopt ``scipy.optimize" in the SciPy library with the method ``Powell". Another challenge is that the number of free parameters scales exponentially with the number of physical qubits $n$. For example, $n=4$ requires $2^8 = 256$ independent parameters to parametrize $U$, and this number becomes $1024$ for $n=5$. Another problem of having a large $n$ is that the evaluation of the forward dynamics $e^{{\cal L} \delta t}$. Using the standard vectorisation method \cite{Zwolak04}, this will involve the integration of the Liouville super-operator matrix with size $2^{2n} \times 2^{2n}$. Due to these technical limitations, the optimisation problem can only be done practically for a small system size $N \lesssim 6$.

\newpage
\subsection{D. Encoding two-logical qubits into multiple-physical qubits}
Based on the optimization protocol, we construct the code space that encodes two-logical qubits into $n$-physical qubits. We study to different noise models, the bit-flipping noise $\Gamma_X \sum_{i=1}^n{\cal D} [\sigma_x^{(i)}]$ and the amplitude damping noise $\Gamma_- \sum_{i=1}^n{\cal D} [\sigma_-^{(i)}]$. We perform a numerical optimization of the code space by following Eq.~\eqref{eq:Suppl_opt_prob}, and construct the recovery dynamics ${\cal L}_B$ using Eqs.~\eqref{eq:Suppl_P_rev_L} and \eqref{eq:Suppl_P_rev_H} based on the optimized code space.

Figure~\ref{fig:Suppl_opt_basis} shows that after optimizing the code space, a higher entanglement fidelity can be obtained after encoding and applying the recovery protocol. We also observe that there seems to be a threshold value of $\beta$ to get a higher entanglement fidelity by increasing the number of physical qubits $n$.  For the bit-flipping noise, $n\geq4$ provides an efficient recovery of the encoded qubits, while $n\geq5$ is required for the amplitude damping noise.
\begin{figure}[ht]
\includegraphics[width=0.45\linewidth]{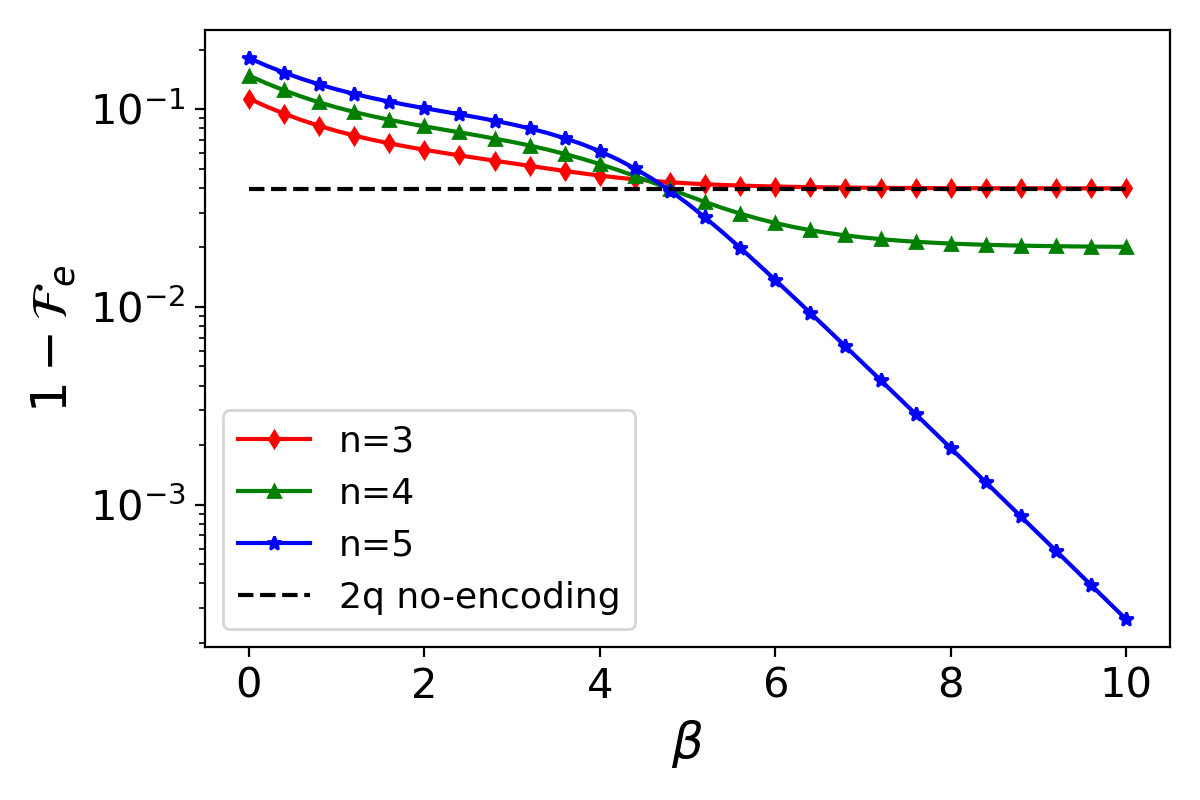}
\includegraphics[width=0.45\linewidth]{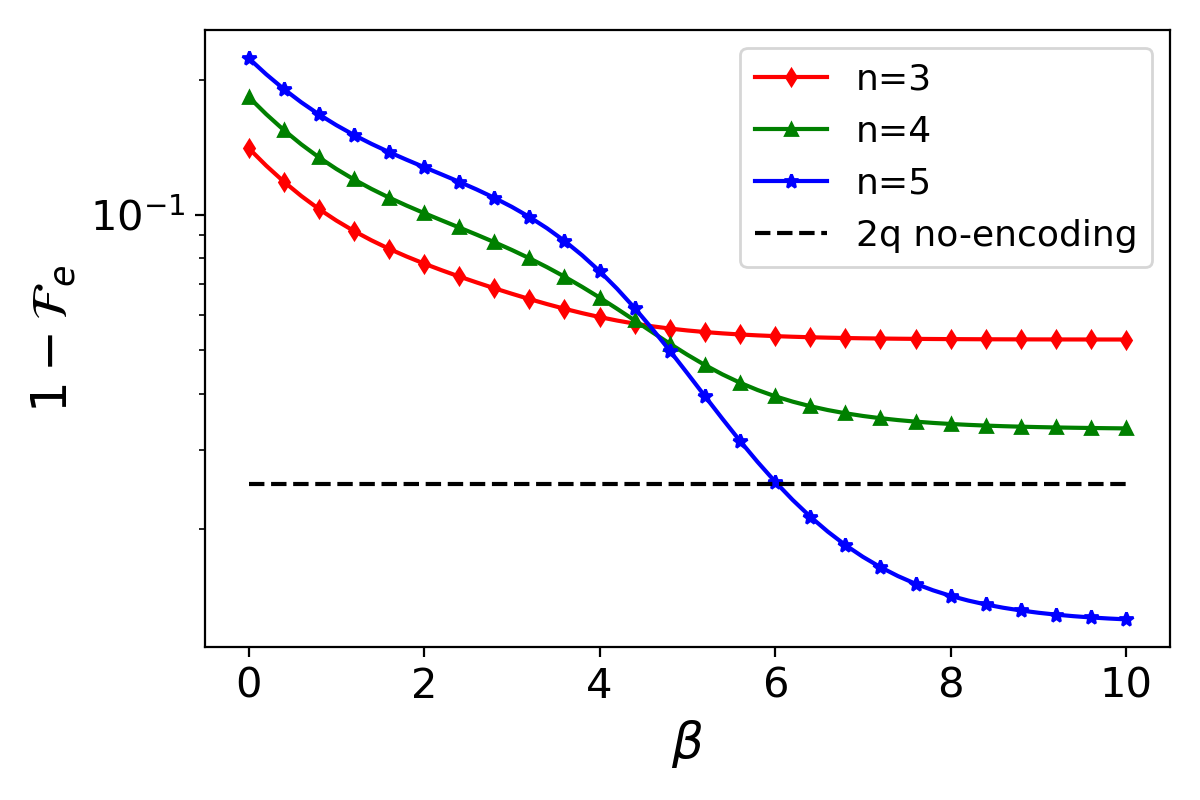}
\caption{Error rate in terms of the entanglement fidelity $1-{\cal F}_e$ after applying the recovery dynamics ${\cal L}_B$ to $n$-physical qubits which encodes two-logical qubits for (a) The bit-flipping noise $\Gamma_X \sum_{i=1}^n{\cal D} [\sigma_x^{(i)}]$ with rate $p_X = (1- e^{-2\Gamma_X t})/2 = 0.02$ and (b) the amplitude damping noise $\Gamma_- \sum_{i=1}^n{\cal D} [\sigma_-^{(i)}]$ with rate $p_- = 1-e^{-\Gamma_-t}=0.05$.}
\label{fig:Suppl_opt_basis}
\end{figure}

We also show that the time-dependent Hamiltonian dynamics of the two logical qubits can be more efficiently simulated by suppressing the noise via the recovery dynamics. To show this, we consider the optimized code space ${\cal C}$ for two-logical qubits encoded in five-physical qubits ($n=5$) against the amplitude damping noise $\Gamma_- \sum_{i=1}^5{\cal D} [\sigma_-^{(i)}]$. Let us suppose that a quantum state, initially prepared at $\ket{00}_L$ evolves under the Hamiltonian, 
$$
H(t) = 2 \sin(7 t) X_L \otimes X_L + 1.4 \sin(3 t) Z_L \otimes Z_L + 2 \cos(10 t) X_L \otimes \mathbb{1} + \mathbb{1} \otimes Z_L,
$$
with the amplitude damping noise $\Gamma_- \sum_{i=1}^5{\cal D} [\sigma_-^{(i)}]$.  The time-evolution of the two-logical qubits under the recovery dynamics can be written as
$$
\dot \rho = {\cal L}_S (\rho)= -i \left[H(t) + \Gamma_- \sum_{i=1}^5 H_C(\gamma,\sigma_-^{(i)}) ,\rho \right] + \Gamma_- \sum_{i=1}^5{\cal D} [\sigma_-^{(i)}](\rho) + \Gamma_- \sum_{i=1}^5{\cal D} [\gamma^{\frac{1}{2}} \sigma_-^{(i)} \gamma^{-\frac{1}{2}} ](\rho),
$$
by taking $\gamma = e^{-\beta Q}/{\rm Tr}[e^{-\beta Q}]$ with $Q = -P_{\cal C}$.
We then compare the state under the recovery dynamics
$$
\rho_\tau = {\cal T} \left[ e^{\int_0^\tau {\cal L}_S dt'} \right] (\ket{00}_L\bra{00}),
$$
to the noise-free evolution,
$$
\rho_\tau^{\rm ideal} = U(\tau) (\ket{00}_L\bra{00}) U^\dagger(\tau),
$$
where $U(\tau) = {\cal T} \left[ e^{-i \int_0^\tau H(t) dt' }\right]$.

As  a two-qubit state can be expressed in terms of two-qubit Pauli operators as $\rho = \frac{1}{4} \sum_{i,j}^{0,1,2,3} T_{ij} \left( \sigma_i \otimes \sigma_j \right)$ with $\sigma_0 = \mathbb{1}$, the dynamics of a logical two-qubit state can be fully characterized by evaluating the expectation values of $\langle O_1 \otimes O_2 \rangle $ with $O_{1,2} \in \{ \mathbb{1}_L, X_L, Y_L, Z_L \}$. Figure~\ref{fig:Suppl_QEC_2q_all} shows that all the expectation values of the encoded logical qubit becomes closer to the ideal (noise-free) case after applying the recovery dynamics ${\cal L}_B$, and the fidelity between the noise-free state becomes closer to $1$.
\begin{figure}[ht]
\includegraphics[width=0.23\linewidth]{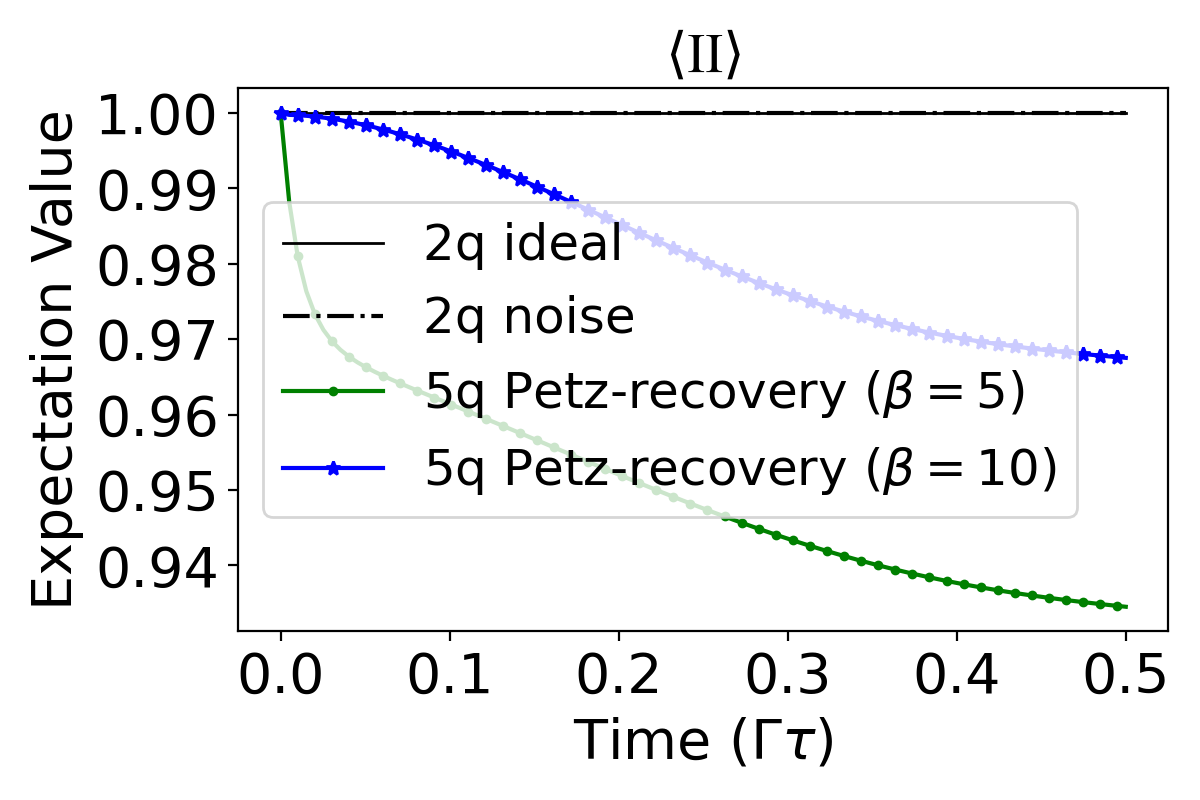}
\includegraphics[width=0.23\linewidth]{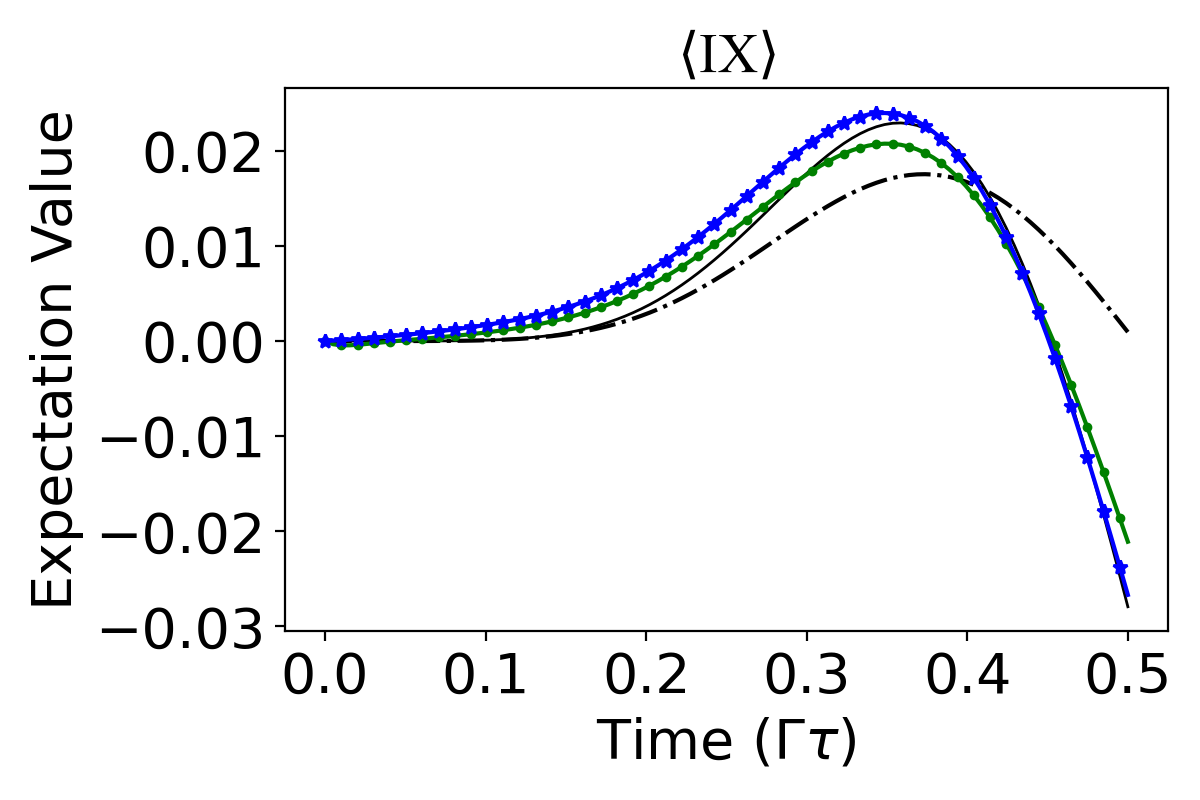}
\includegraphics[width=0.23\linewidth]{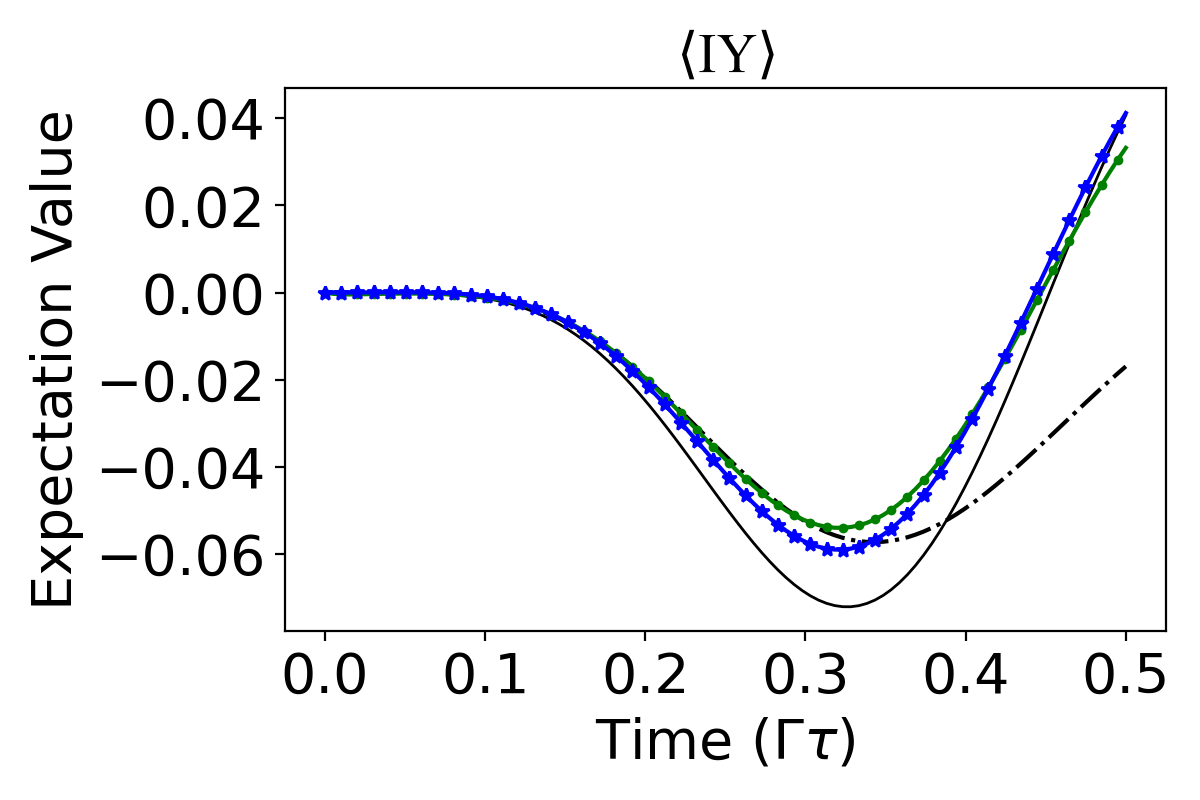}
\includegraphics[width=0.23\linewidth]{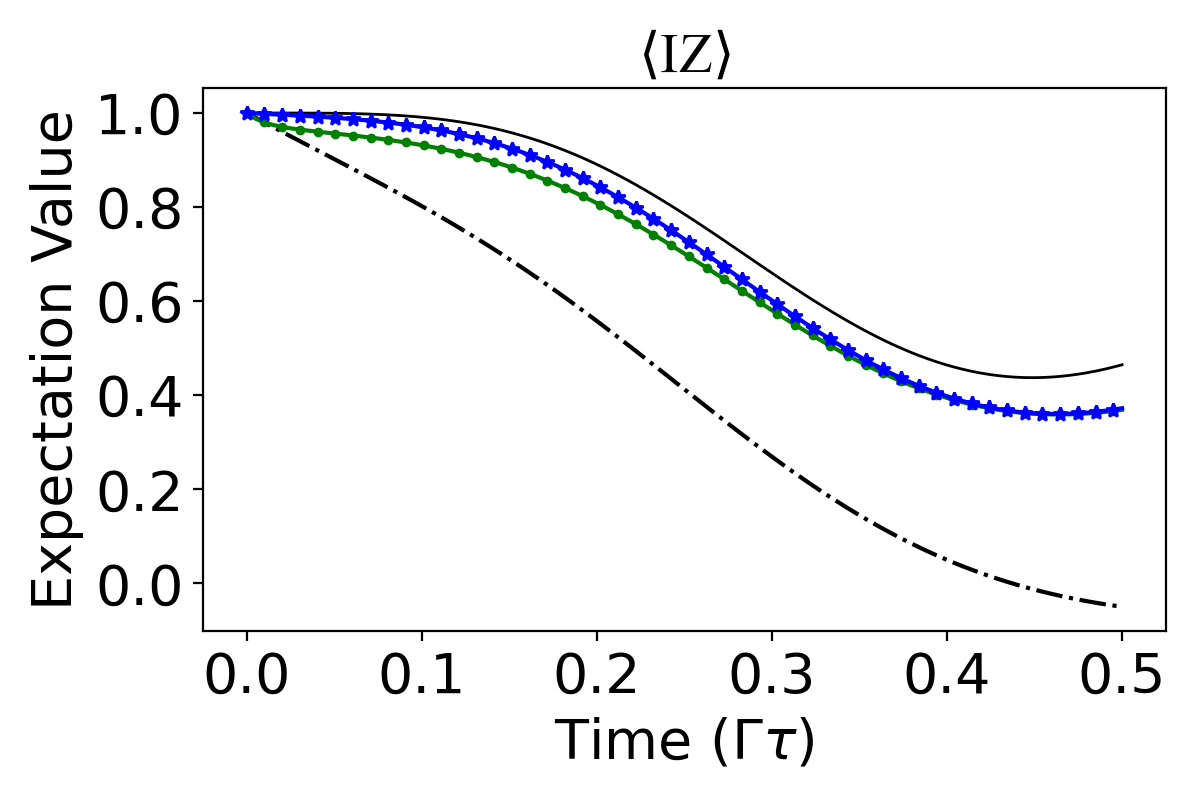}
\includegraphics[width=0.23\linewidth]{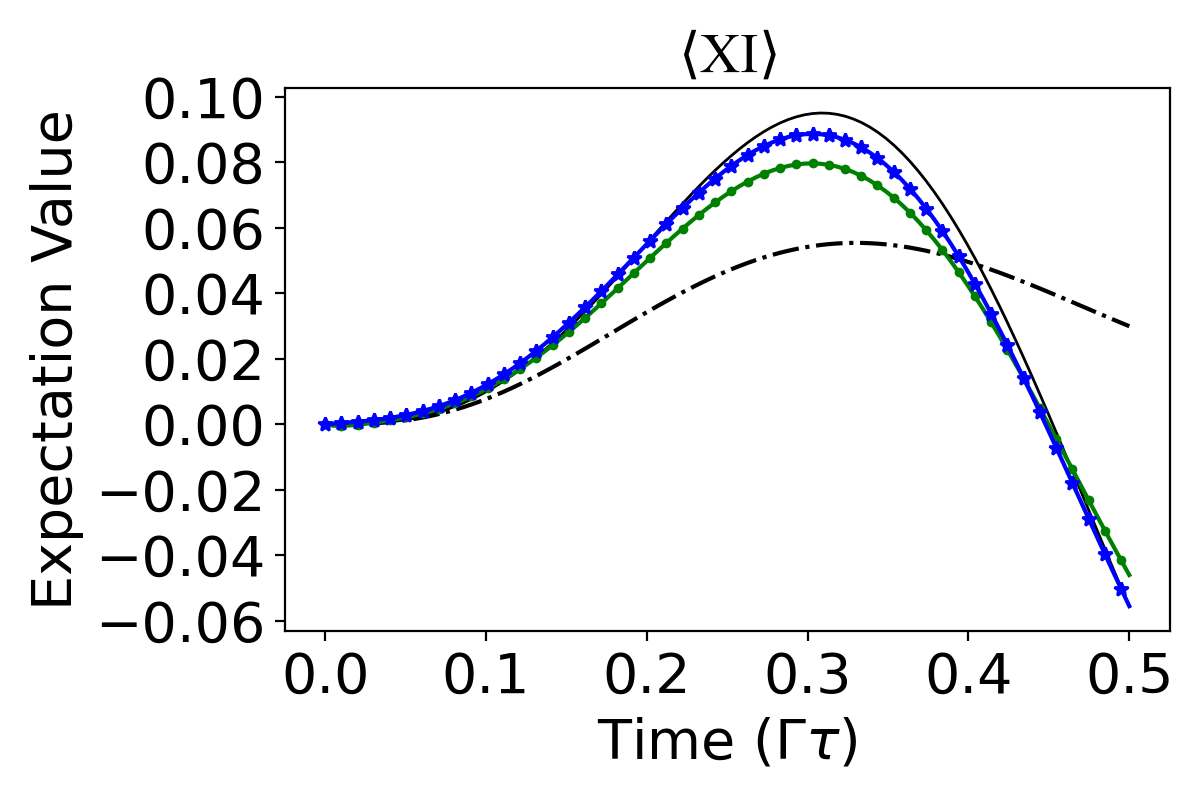}
\includegraphics[width=0.23\linewidth]{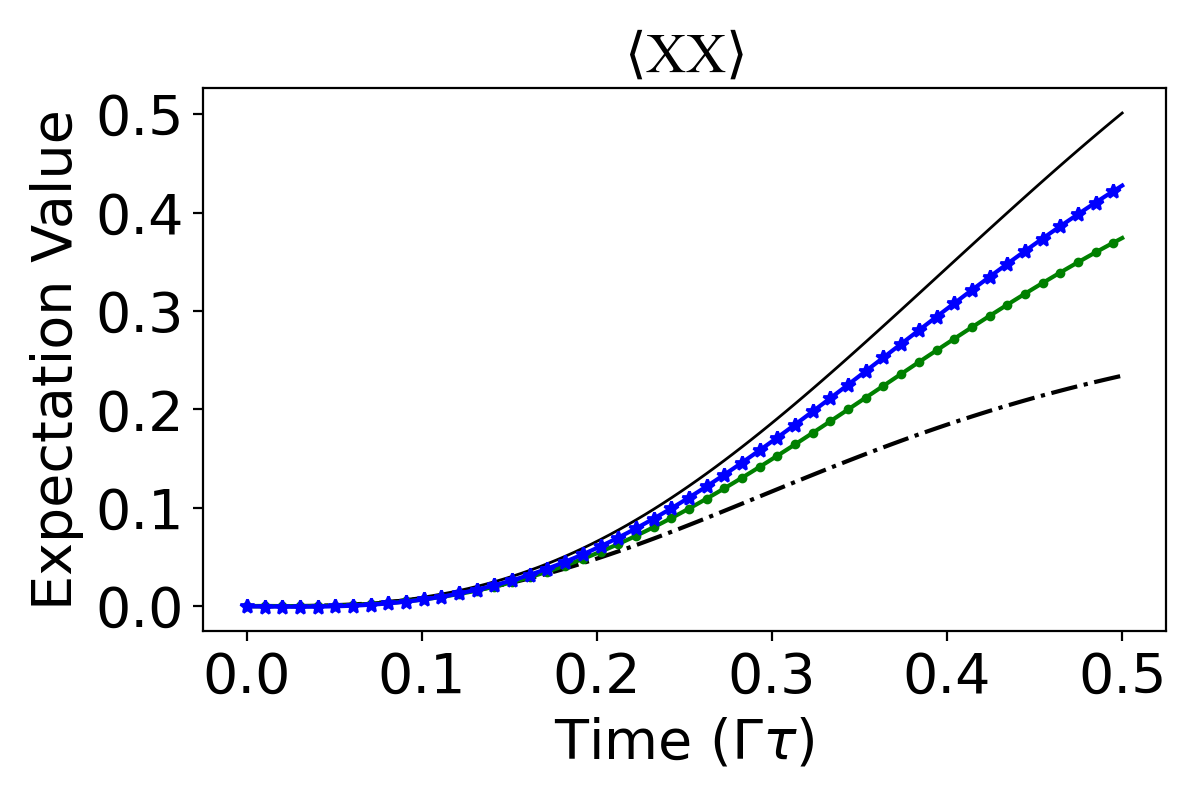}
\includegraphics[width=0.23\linewidth]{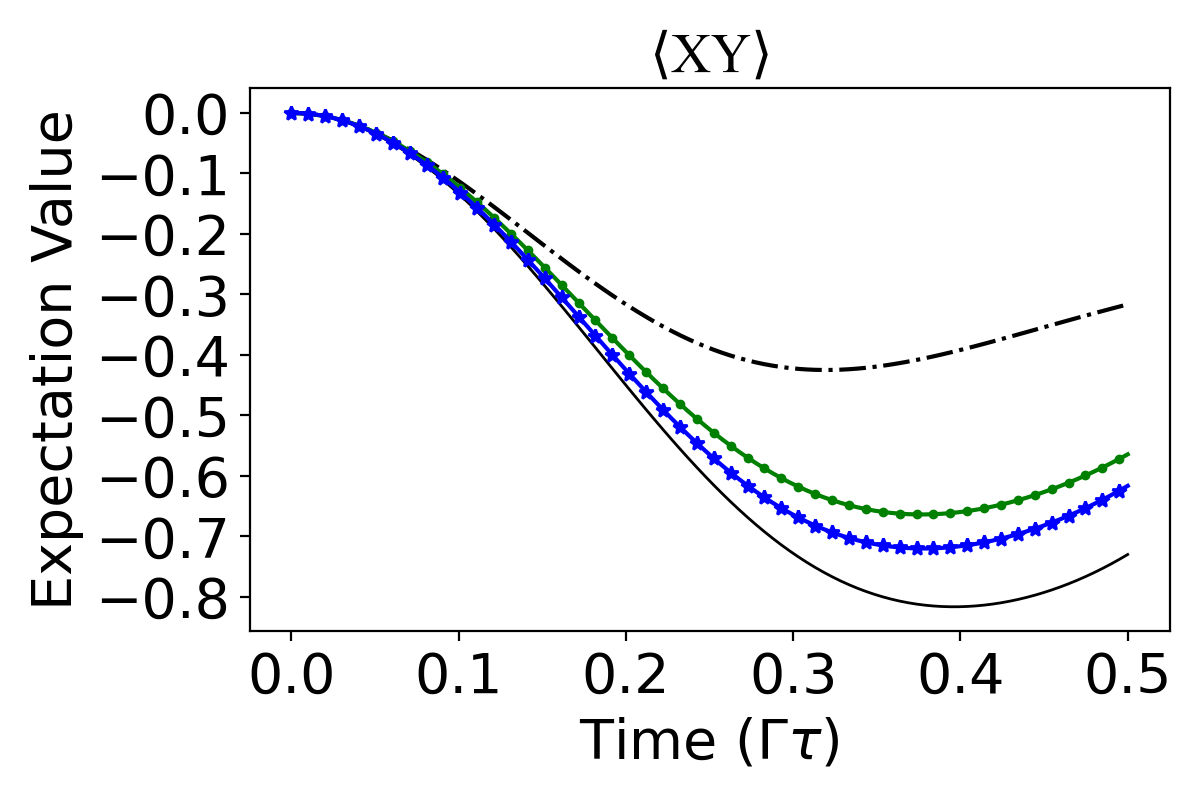}
\includegraphics[width=0.23\linewidth]{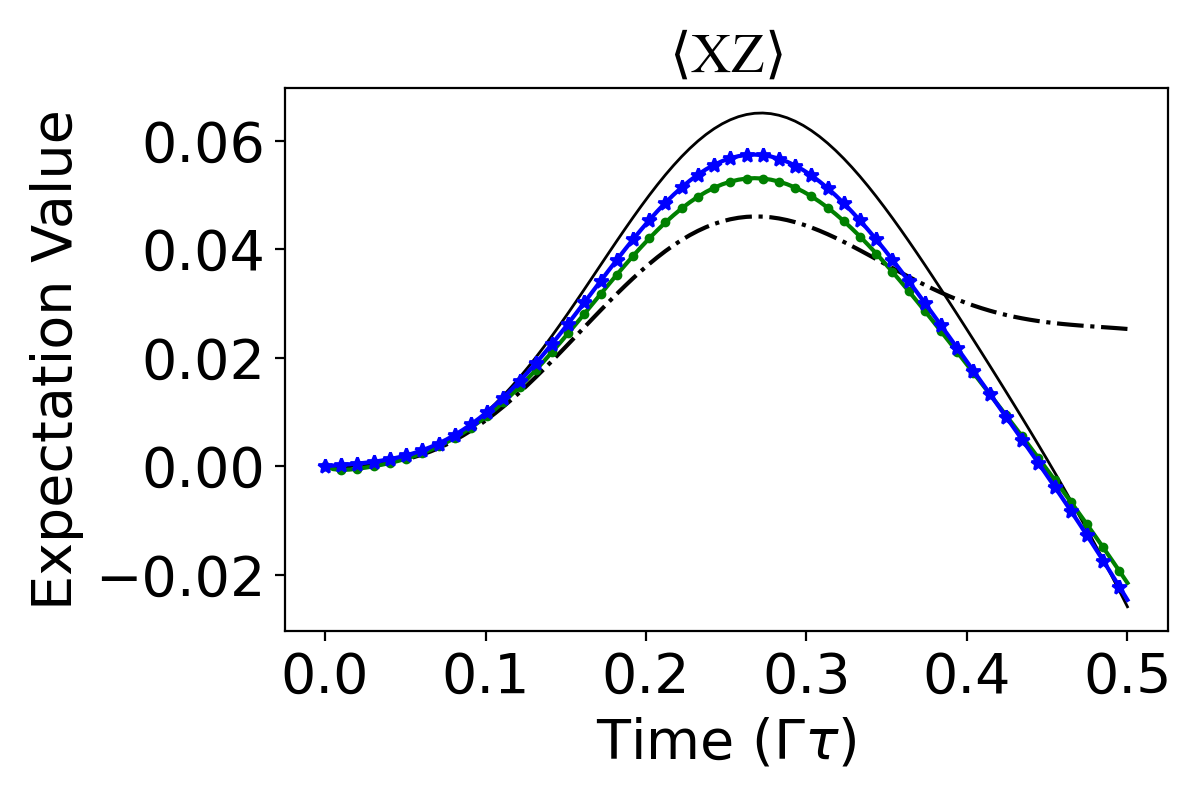}
\includegraphics[width=0.23\linewidth]{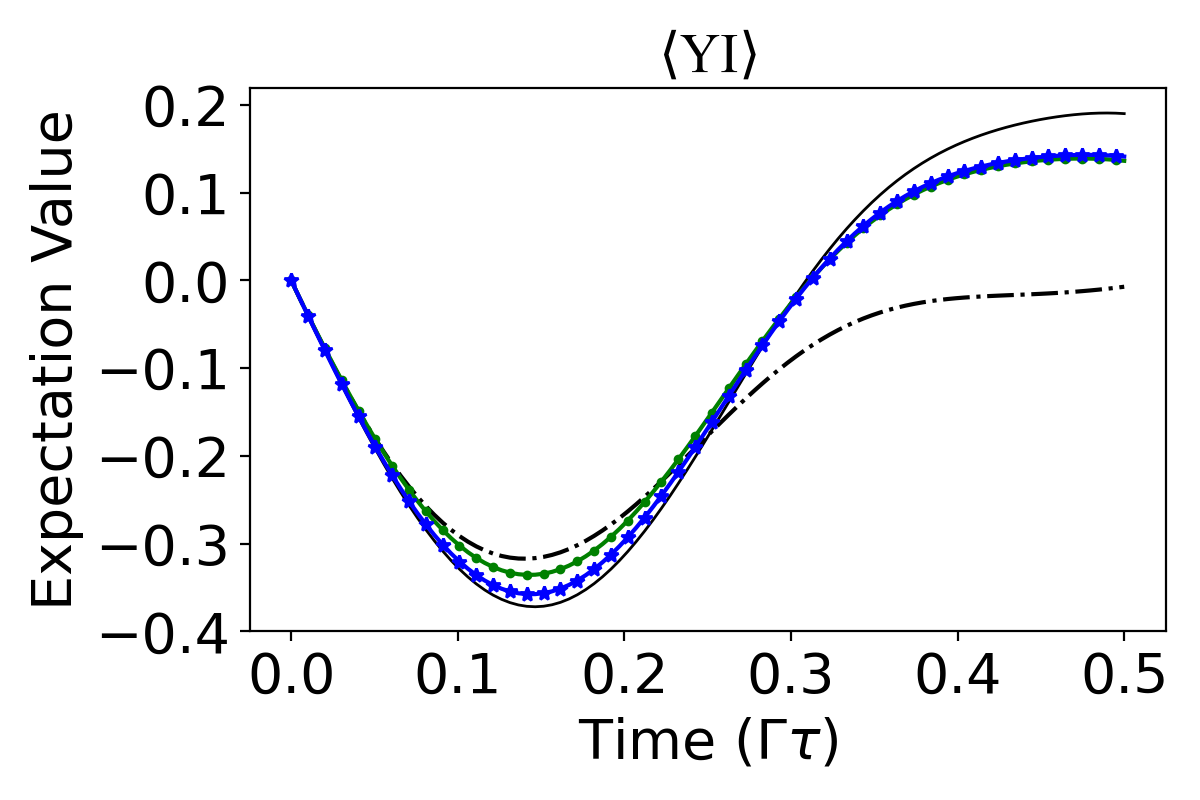}
\includegraphics[width=0.23\linewidth]{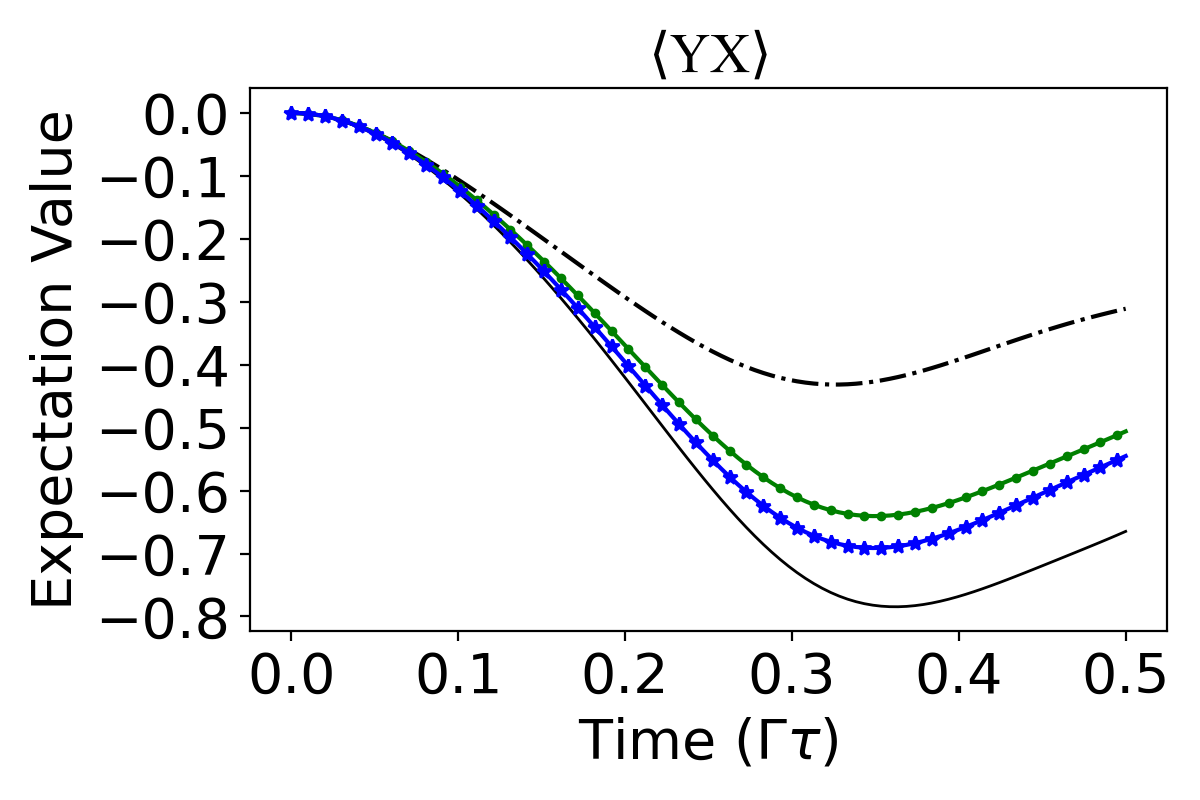}
\includegraphics[width=0.23\linewidth]{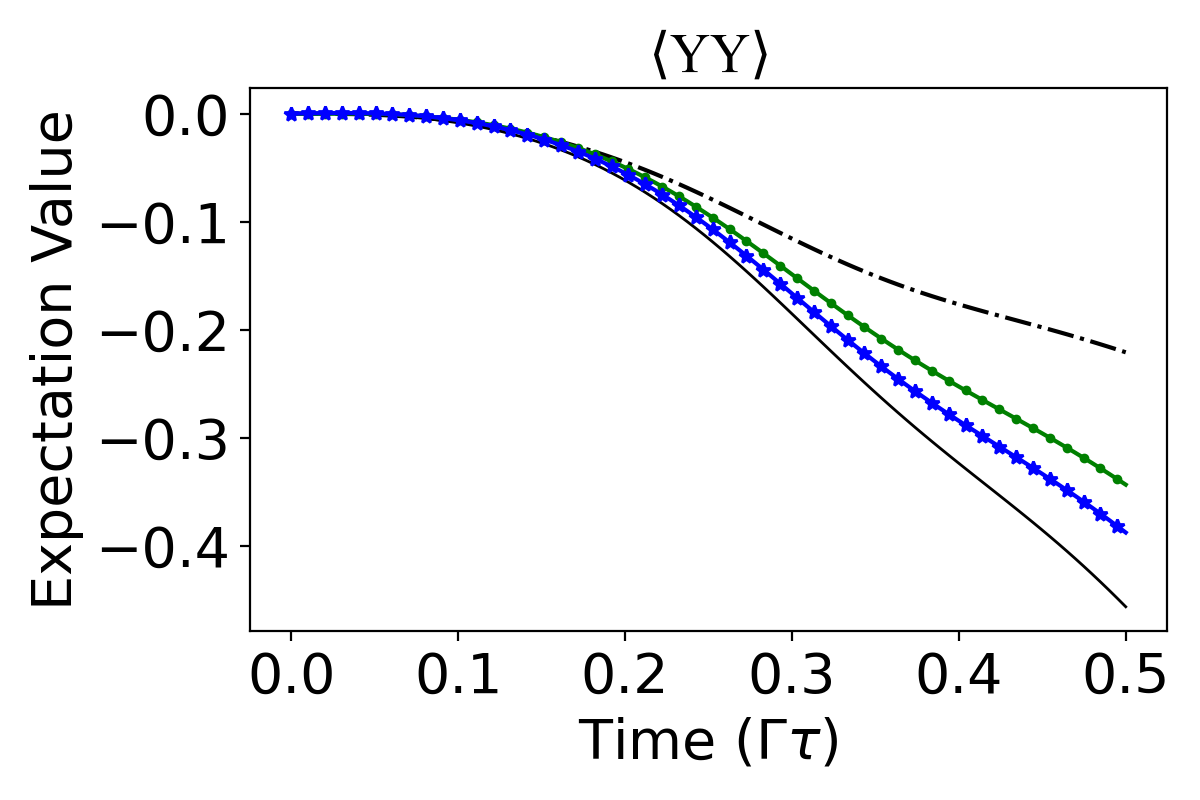}
\includegraphics[width=0.23\linewidth]{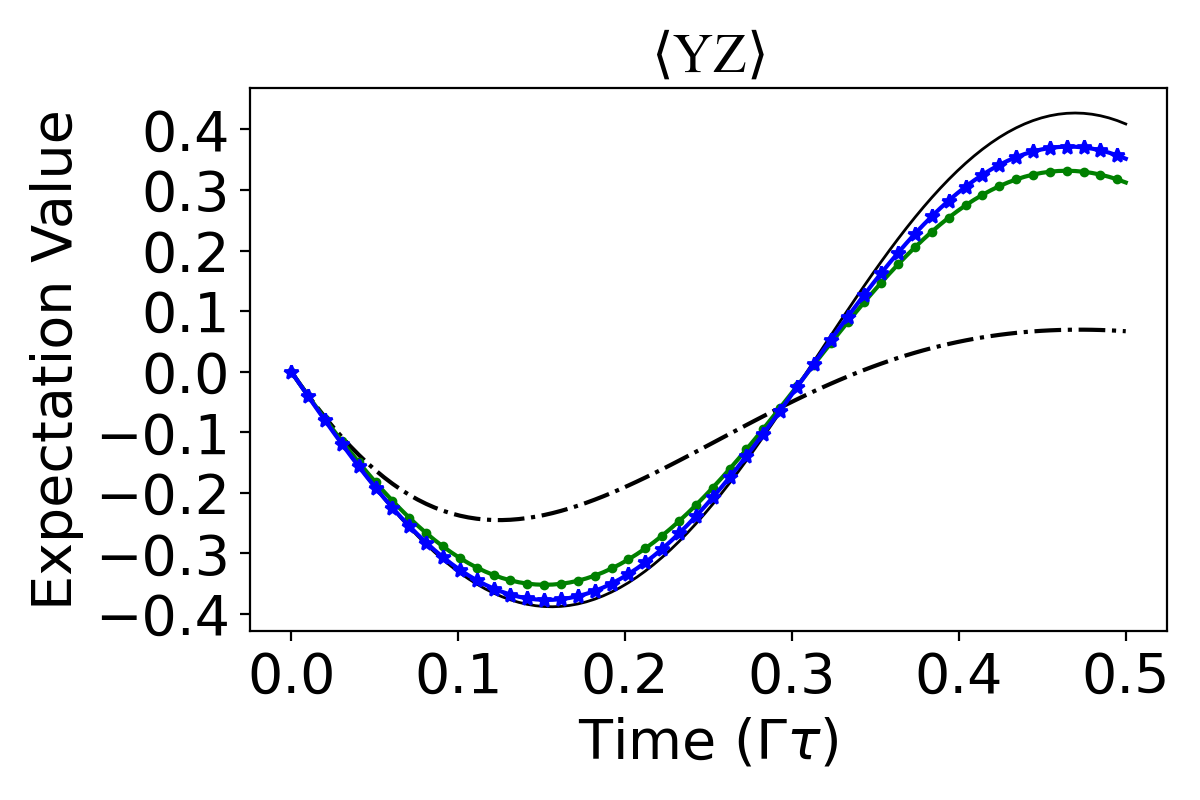}
\includegraphics[width=0.23\linewidth]{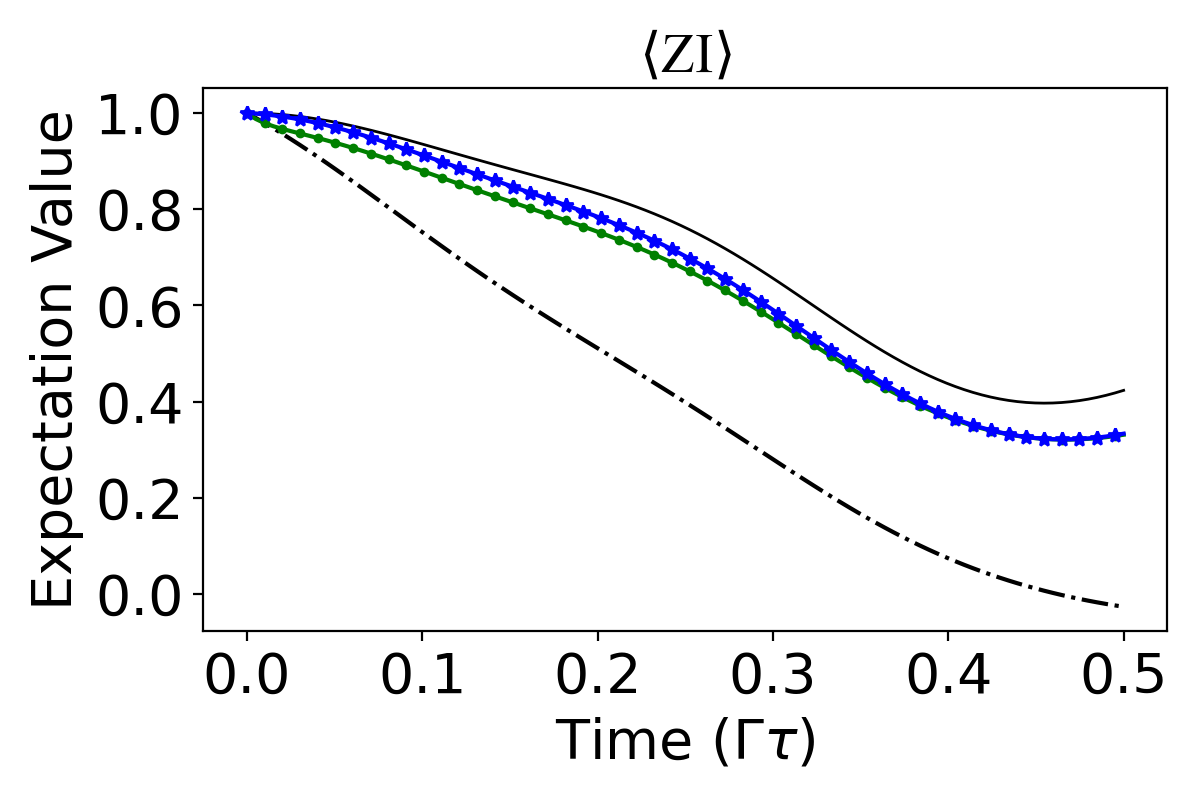}
\includegraphics[width=0.23\linewidth]{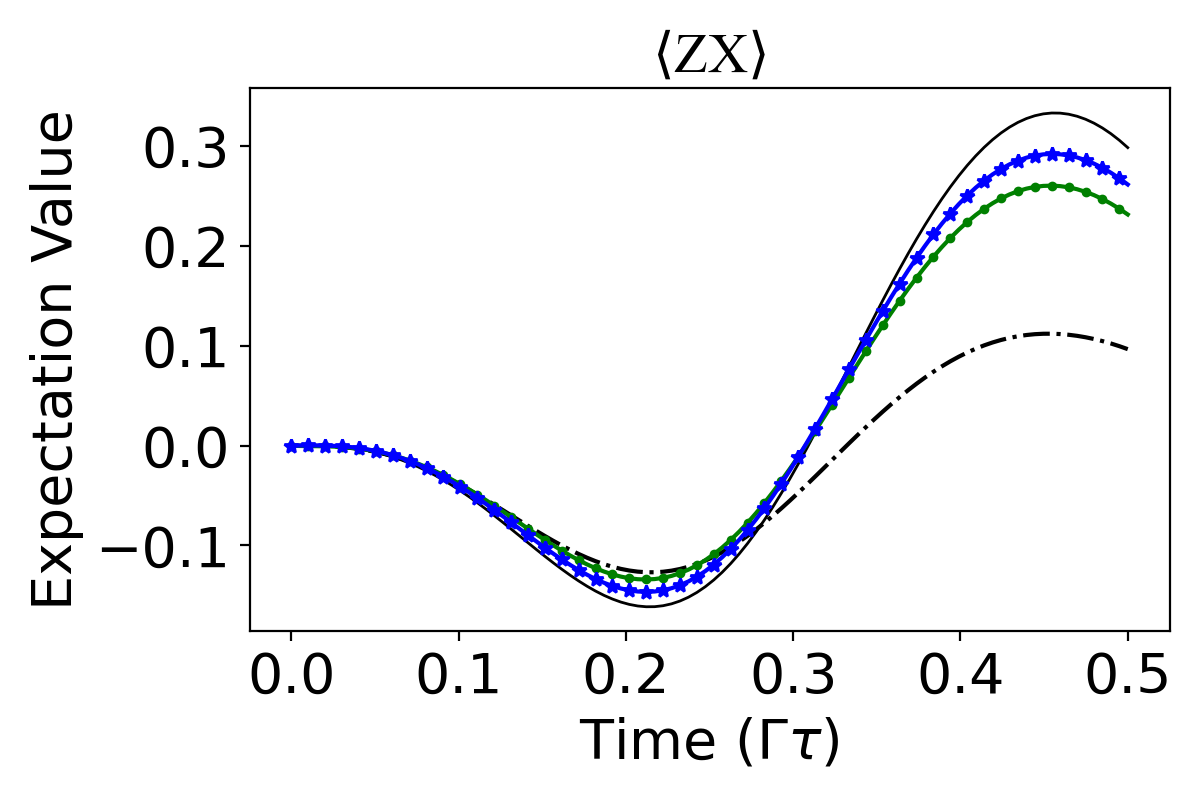}
\includegraphics[width=0.23\linewidth]{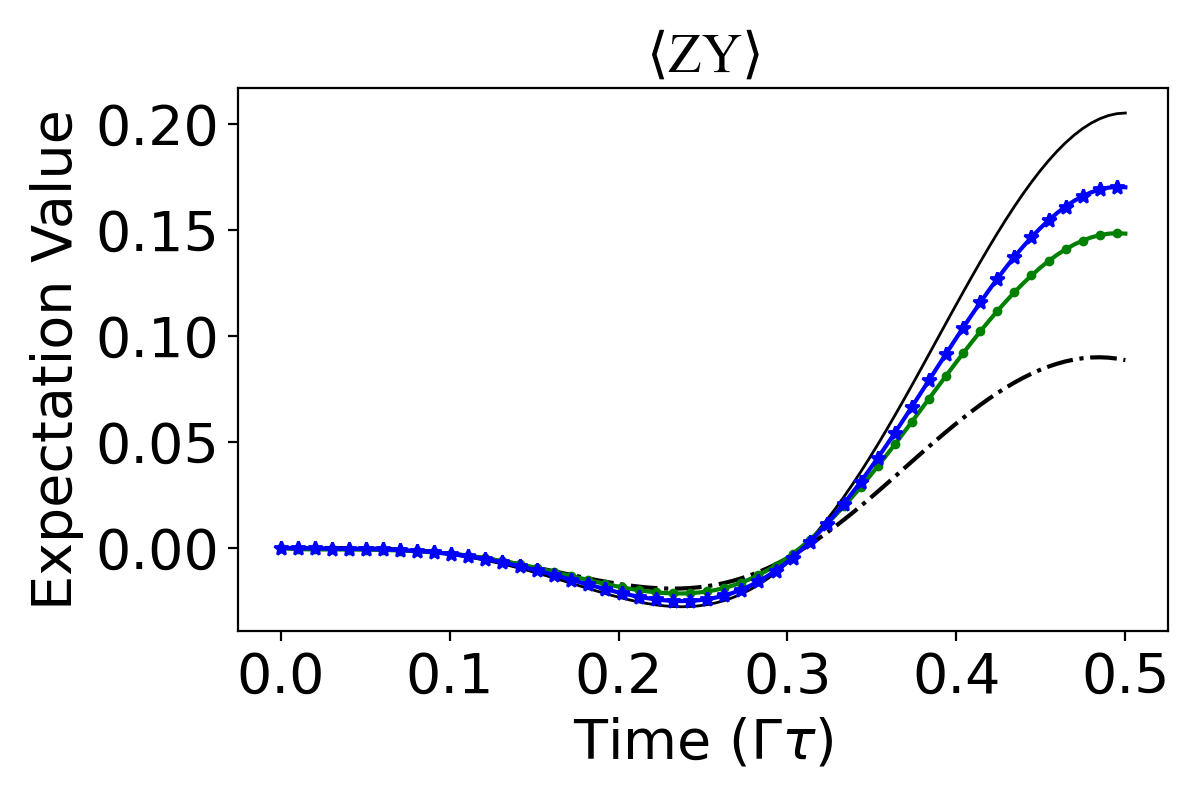}
\includegraphics[width=0.23\linewidth]{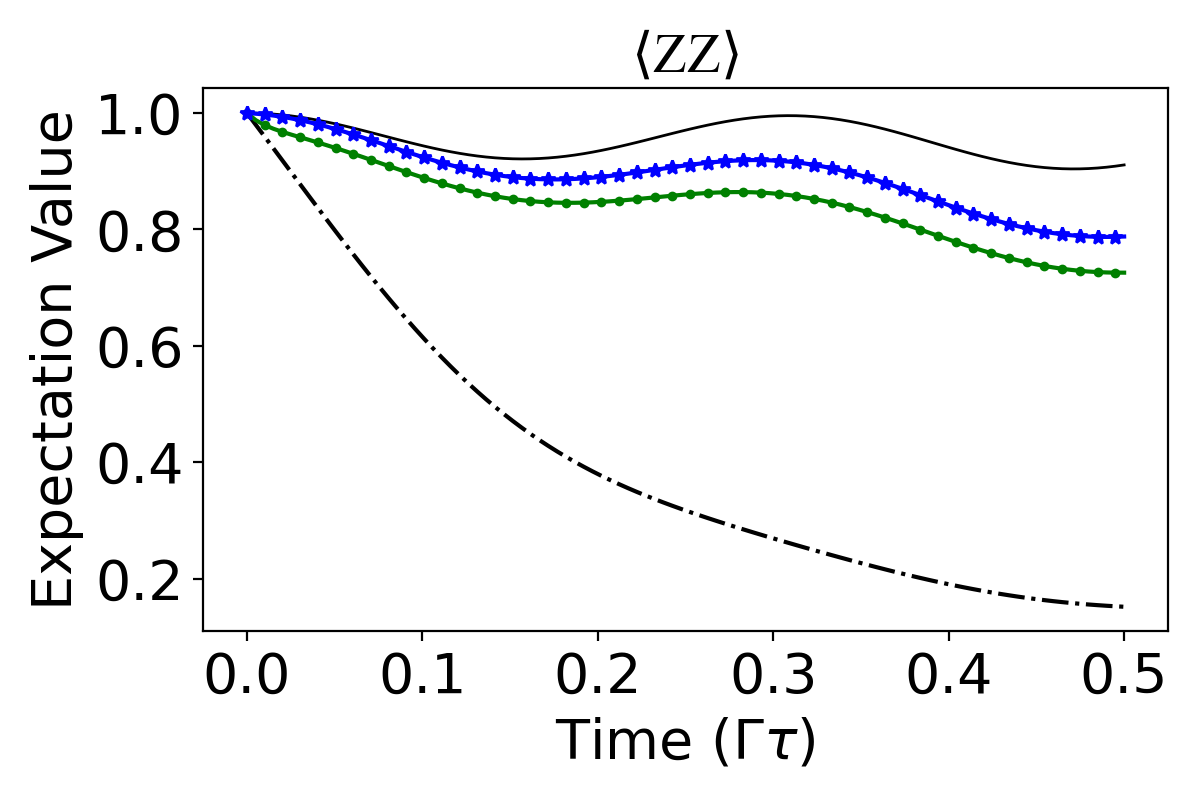}
\includegraphics[width=0.42\linewidth]{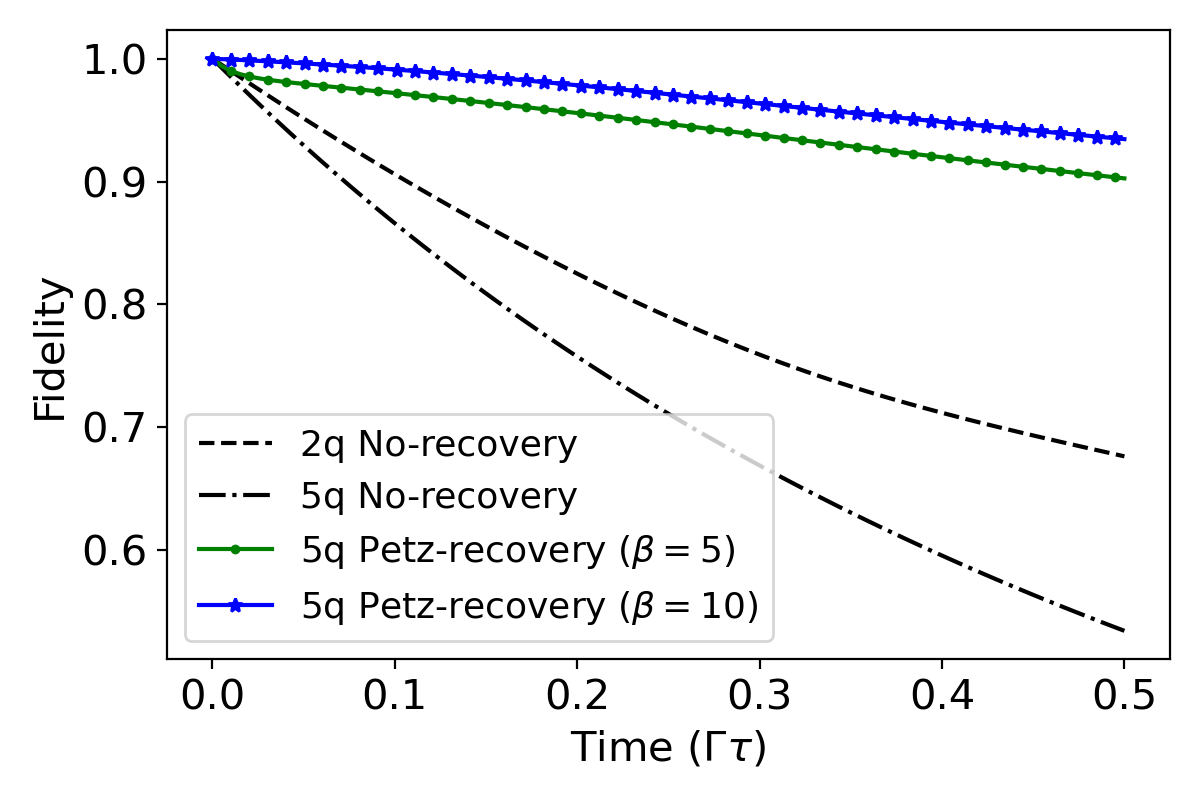}
\caption{Evolution of a two-qubit logical state encoded in five physical qubits and its fidelity between the noise-free state. The asterisk points represent the expectation values after recovery, while solid and dashed lines represent those for noise-free and noisy dynamics, respectively. Dot-dashed lines refer to a noisy two-qubit dynamics without encoding.}
\label{fig:Suppl_QEC_2q_all}
\end{figure}




\end{document}